\documentclass[10pt,a4paper,twoside]{amsart}

 \usepackage[english]{babel}

\usepackage{amsthm}
\usepackage{amssymb}
\usepackage{amsmath}
\usepackage{amsopn}
\usepackage{bm}
\usepackage{hhline}
\usepackage{graphicx}

\usepackage{cite}
\allowdisplaybreaks[4]

\newtheorem{lemma}{Lemma}
\newtheorem{theorem}{Theorem}
\newtheorem{proposition}{Proposition}
\newtheorem{corollary}{Corollary}

\theoremstyle{definition}
\newtheorem{remark}{Remark}
\newtheorem{example}{Example}

\DeclareMathOperator{\Complex}{\mathbb{C}}
\DeclareMathOperator{\Real}{\mathbb{R}}
\DeclareMathOperator{\Res}{Res}
\DeclareMathOperator{\Mat}{Mat}
\DeclareMathOperator{\diag}{diag}

\DeclareMathOperator{\const}{const}

\newcommand{\rme}{\mathrm{e}}

\DeclareMathOperator{\rank}{rank}
\DeclareMathOperator{\spanOp}{span}

\newcommand{\dd}{\mathrm{d}}
\DeclareMathOperator{\Rew}{Re}
\DeclareMathOperator{\Imw}{Im}

\allowdisplaybreaks[4]

\title{On degenerate sigma-functions in genus two}
\author{Julia Bernatska}
\address{National University of `Kyiv-Mohyla Academy', Kyiv, Ukraine}
\email{bernatskajm@ukma.kiev.ua, jbernatska@gmail.com}

\author{Dmitry Leykin}
\address{NASU Institute of Magnetism, Kyiv, Ukraine}
\email{dmitry.leykin@gmail.com}

\begin{document}

\begin{abstract}
We obtain explicit expressions for genus 2 degenerate sigma-func\-tion in terms of
genus $1$ sigma-function and elementary functions as solutions
of a system of linear PDEs satisfied by the sigma-function. By way of application
we derive a solution for a class of generalized Jacobi inversion problems on elliptic curves, 
a family of Schr\"{o}dinger-type operators on a line with common spectrum consisting of 
a point and two segments, explicit construction of a field of three-periodic meromorphic functions.
Generators of rank $3$ lattice in $\Complex^2$ are given explicitly.
\end{abstract}


\maketitle


\section*{Introduction}
The concept of a sigma-function in higher genus was introduced
by F. Klein \cite{Klein} in 1886 as an extensive generalization of elliptic 
Weierstrass sigma-function \cite{Weier1894}.
The importance of a sigma-function lies in the fact that it is a convenient generator of Abelian functions
in $g$ complex variables, i.e. meromorphic multiply periodic functions that possess 
the maximal number $2g$ of periods. From this viewpoint sigma-function in genus 2 
was studied since the time of Klein, and the classical results were 
well documented in Baker monograph \cite{Baker}.

Theory of sigma-functions progresses in several ways. 
The first approach considers a generalization called in \cite{BEL1997} 
Kleinian sigma-function $\sigma(u)=
\exp\big(\frac{1}{2} u^t \eta \omega^{-1} u\big) \theta(\omega^{-1} u;\omega,\omega')$, 
which is a modular invariant representative of the class of theta 
functions\footnote{Here Riemann $\theta$-function is used. As usual, $\omega$ 
and $\eta$ denote first and second kind integrals along $\mathfrak{a}$-cycles, and $\omega'$ denotes 
first kind integrals along $\mathfrak{b}$-cycles.} with $u\in\Complex^g$. 
Dealing with this expression for the multivariative sigma-function the authors 
have examined fields of Abelian functions 
associated with hyperelliptic curves,
described Jacobi and Kummer varieties as algebraic varieties,  
developed contemporary applications of Abelian functions  
to completely integrable equations of theoretical and mathematical physics. Further theoretical 
developments in study of sigma and its relation to theta and tau functions are achieved 
in  \cite{Nak2010, Nak2015}. Modular definition of sigma-function for generic 
algebraic curves is proposed  in \cite{KorShr} whilst the case of $(n,s)$-curves and, in particular,  
of hyperelliptic curves is studied in this context in \cite{EEE2013}.

Another approach, aimed to construct series expansions for multivariative sigma-functions,
is developed in sequential papers \cite{BEL1999,BL2002,BL2004,BL2008}.
This is a powerful technique for obtaining  sigma-series explicitly as well as 
an elegant theory derived from unfolding Pham singularities. The theory gives many byproducts. 
For example, the canonical basis of second kind differentials associated to the first kind differentials are
constructed without introducing Kleinian bi-differential unlike the first approach; 
in their turn the second kind differentials can be used for constructing the bi-differential.
Algebraic identities and addition laws in the field of Abelian functions are easily obtained 
from a sigma-series expansion, as well as Hirota bilinear equations 
for integrable systems associated with a curve.
The theory is applicable for both hyper- and non-hyperellitic plane curves.

The equivariant approach based on covariance with respect to transformations 
of an algebraic curve is proposed in \cite{AEE2003,A2008,A2011}. For 
hyperelliptic curves the transformations are induced by SL(2,$\Complex$) action.
Identities for multivariative Abelian functions arise as
finite-dimensional irreducible representations of the corresponding algebra $\mathfrak{sl}(2,\Complex)$; 
and the covariant form of Kleinian bi-differential is constructed. It was shown that the Hirota derivative 
plays the role of a partial intertwining operator in the representation theory of $\mathfrak{sl}(2,\Complex)$
and so can serve for generation of algebraic invariants \cite{A1999}.  A generalization of bilinear Hirota
operators possessing an equivariance property is applied for constructing a basis in the space
of Abelian functions with poles of at most a given order \cite{EA2012}.

There are also computational approach where series expansions for multivariative sigma-functions 
in higher genera and algebraic identities between the corresponding Abelian functions
are constructed by numerical methods independently of the second approach,
but partly inspired by it. Series for some sigma-functions in genera $3$, $4$, $6$  
were obtained \cite{BG2006, EEMOP2007, 
BEGO2008, EEG2009, EEG2010} as well as identities for
Abelian functions of the kind $\wp_{[k]}(u)=-\partial^{[k]} \log \sigma(u)$.

\medskip
The present paper is essentially based on the theory of constructing series for
multivariative sigma-functions, mostly on the paper \cite{BL2004}. The theory 
originates from  Weierstrass's definition \cite{Weier1894} of sigma-function as
the entire function depending on 
the three variables $(u;g_2,g_3)\in \Complex \times\Complex^2$  and satisfying the 
set of differential equations
\begin{gather}
Q_0 (\sigma) = 0,\qquad Q_2 (\sigma) = 0, \label{WeierAnOP}\\
\begin{split}
 &Q_0(\sigma) = -u \sigma_u + 4g_2 \sigma_{g_2} + 6g_3 \sigma_{g_3} + \sigma , \notag\\
 &Q_2(\sigma) = -\tfrac{1}{2}\sigma_{u,u} - \tfrac{1}{24} g_2 u^2 \sigma 
 + 6 g_3 \sigma_{g_2} + \tfrac{1}{3} g_2^2 \sigma_{g_3} \notag
 \end{split}
\end{gather}
with initial condition $\sigma(u;0;0)=u$;
here $g_2$ and $g_3$ are parameters of Weierstrass elliptic curve $y^2=4x^3-g_2 x-g_3$.
A multivariative sigma-function in genus $g$ is an entire function $\sigma(u;\lambda)$ 
of $3g-m$ complex variables $(u;\lambda)\in \Complex^g \times\Complex^{2g-m}$, where $m$ is modality
(in the hyperelliptic case, we are dealing with in the present paper, $m=0$, for more details see \cite{BL2008}).
The multivariative sigma-function is defined by a set of heat equations \cite{BL2004} 
similar to \eqref{WeierAnOP}
with initial conditions in the form of so called Schur-Weierstrass polynomials \cite{BEL1999}.

\medskip
In this paper we study cases when sigma-function in genus $2$ can be represented 
as an aggregate of sigma-function in genus $1$ and elementary functions.  For $g=2$ sigma-function
depends on six variables $(u;\lambda)\in \Complex^2 \times\Complex^4$, 
where $\lambda$ is the set of parameters of a genus $2$ curve
\begin{equation}\label{CurveHG2Intro}
y^2 = x^5 + \lambda_4 x^3 + \lambda_6 x^2 + \lambda_8 x + \lambda_{10}.
\end{equation}

In fact, curve \eqref{CurveHG2Intro} has \emph{actual genus} $2$ 
only if a certain constraint is imposed on $\lambda$.
We say that the curve has \emph{virtual genus} $2$ and focus on the cases when its actual genus is lower.
That is the cases when genus $2$ sigma-function can be expressed in terms of genus $1$ sigma-function 
and elementary functions. Strata of the space of parameters $\lambda$ corresponding to fixed actual genera 
of \eqref{CurveHG2Intro} are analyzed in Sections~\ref{s:Strata} and \ref{s:Frames}. 
Then we proceed to our main question by carrying out an analysis of the system of 
linear partial differential equations, so called heat equations in a non-holonomic frame \cite{BL2004}, 
that are satisfied by sigma-function in Section~\ref{s:HeatEqs}, and 
derive our main result in Section~\ref{s:DegenSigma}. 

In Section~\ref{s:Appl} we apply the result to a few selected problems:
solution of a generalized Jacobi inversion problem, construction of 
a Schr\"{o}dinger type operator with 
spectrum composed of two segments and a point, description of the structure of a field
of three periodic functions in two complex variables. In the course of our research we
also obtain a stratification of the space of parameters of genus $2$ curves \eqref{CurveHG2Intro}
with respect to the rank of a period lattice corresponding to the curve.

\section{Notation}
Below in the paper we consider the space $\mathcal{C}$ of genus $2$ curves of the form 
\begin{equation}\label{CurveHG2}
x^5 - y^2 + \lambda_4 x^3 + \lambda_6 x^2 + \lambda_8 x + \lambda_{10}=0.
\end{equation}
The parameters $\lambda=(\lambda_4,\lambda_6,\lambda_8,\lambda_{10})$ run over $\Complex^4$.
Genus 2 sigma-function is denoted by $\bm{\sigma}(u;\lambda)$, where $u=(u_3,u_1)$.
We assign Sat\={o} weights to the variables by the rule $\deg \lambda_i=i$, $\deg u_k=-k$.
Accordingly, with $\deg x=2$ and $\deg y=5$ the left hand side of \eqref{CurveHG2} is homogeneous of
weight $10$. It is important, that most of functions and operators appearing below respect Sat\={o} weights,
in particular, $\deg \bm{\sigma}(u;\lambda)={-}3$.

In what follows we also deal with a family of genus $1$ curves
\begin{equation}\label{CurveHG1}
X^3 - Y^2 + \gamma_4 X + \gamma_6 =0,
\end{equation}
here $\deg X=2$ and $\deg Y=3$. To avoid confusion we denote the corresponding genus 1 sigma-function by
$\sigma(u_1)$, which stands for standard Weierstrass sigma-function with invariants
$(g_2,g_3)=(-4\gamma_4,-4\gamma_6)$.

For brevity we use the notation $\partial_x$ in the place of $\partial/\partial x$.

\section{Stratification of the space of parameters}\label{s:Strata}
The space of parameters $\Lambda$ is naturally stratified into three strata: $\Lambda_2$, 
$\Lambda_1$, and $\Lambda_0$ which correspond to curves in genus $g\,{=}\,2$, $1$, and $0$ respectively.
\begin{proposition}\label{P:LabmdaS}
The space $\Lambda$ is a disjoint union 
$\Lambda_2 \cup \Lambda_1 \cup \Lambda_0$ and
\begin{align*}
 &\Lambda_2 = \{\lambda\in \Complex^4 \mid \Delta(\lambda)\neq 0\},\\
 &\Lambda_1 = \{\lambda\in \Complex^4 \mid \Delta(\lambda) = 0,\ \Gamma(\lambda)\neq 0\},\\
 &\Lambda_0 = \{\lambda\in \Complex^4 \mid \Gamma(\lambda)= 0\},
\end{align*}
where 
\begin{equation}\label{DeltaDef}
 \begin{split}
  \Delta(\lambda) &= 3125 \lambda_{10}^4 - 3750 \lambda_{10}^3 \lambda_6 \lambda_4
    +2000 \lambda_{10}^2 \lambda_8^2 \lambda_4 +2250 \lambda_{10}^2 \lambda_8 \lambda_6^2 \\ & \quad
   -1600 \lambda_{10} \lambda_8^3 \lambda_6 + 256 \lambda_8^5
     - 900 \lambda_{10}^2 \lambda_8 \lambda_4^3 +825 \lambda_{10}^2 \lambda_6^2 \lambda_4^2
     +560 \lambda_{10} \lambda_8^2 \lambda_6\lambda_4^2    \\ & \quad
     -630 \lambda_{10} \lambda_8 \lambda_6^3 \lambda_4
     +108 \lambda_{10} \lambda_6^5 -128 \lambda_8^4 \lambda_4^2 +144 \lambda_8^3 \lambda_6^2\lambda_4 
     - 27 \lambda_8^2 \lambda_6^4  \\ & \quad +
  \big(108 \lambda_{10}^2 \lambda_4^5
   - 72 \lambda_{10} \lambda_8  \lambda_6 \lambda_4^4
   + 16\lambda_{10} \lambda_6^3 \lambda_4^3  +16 \lambda_8^3 \lambda _4^4 - 4 \lambda_8^2 \lambda_6^2 \lambda_4^3 ,
 \end{split}
\end{equation}
and
\begin{equation*}
  \Gamma(\lambda) = \begin{pmatrix} 
  50 \lambda_{10}\lambda_6 - 80 \lambda_8^2
  + 36 \lambda_8 \lambda_4^2 - 27 \lambda_6^2 \lambda_4 
  - 4 \lambda_4^4 \\
   200 \lambda_{10} \lambda_8  
   - 40 \lambda_{10} \lambda_4^2 - 36 \lambda_8 \lambda_6 \lambda_4 + 27 \lambda_6^3 
   + 4 \lambda_6 \lambda_4^3  \\
    625 \lambda _{10}^2 
    - 720 \lambda_8^2 \lambda_4 + 135 \lambda_8 \lambda_6^2 
    + 308 \lambda_8 \lambda_4^3 - 216 \lambda_6^2 \lambda_4^2 
    - 32 \lambda_4^5 \\ 1600 \lambda_8^3
    - 1040 \lambda_8^2 \lambda_4^2 + 360 \lambda_8 \lambda_6^2 \lambda_4 + 135 \lambda_6^4 
    + 224 \lambda_8 \lambda_4^4 - 88 \lambda_6^2 \lambda_4^3 
   -16 \lambda_4^6
  \end{pmatrix}.
\end{equation*}
\end{proposition}
\begin{proof}
Consider a curve \eqref{CurveHG2} with at least one double point at $(x,y)\,{=}\,(a_2,0)$. It has the form
\begin{equation}\label{CurveHG2Deg}
-y^2 + (x-a_2)^2 \big(x^3 + 2a_2 x^2 + \mu_4 x + \mu_6\big) = 0.
\end{equation}
By subtracting \eqref{CurveHG2Deg} from \eqref{CurveHG2} and collecting coefficients at the power of $x$
we find the following polynomials in $(\lambda _{10},\,\lambda _{8},\,\lambda _{6},\,
\lambda _{4};\,\mu_6,\,\mu_4,\,a_2)$:
\begin{equation}\label{UpsilonDef}
  \Upsilon(\lambda;\mu,a_2) = \begin{pmatrix} 
                     \lambda_4 -\big(\mu_4 - 3 a_2^2\big)\\
    \lambda_6 -\big(\mu_6 - 2 a_2 \mu_4 + 2a_2^3 \big)\\
    \lambda_8 -\big({-} 2a_2 \mu_6 + a_2^2 \mu_4\big) \\ 
    \lambda_{10} - a_2^2 \mu_6
                    \end{pmatrix}.
\end{equation}
The polynomials vanish whenever a curve \eqref{CurveHG2} has the form \eqref{CurveHG2Deg},
that is the curve has genus not greater than 1, equivalently $\lambda\,{\in}\,\Lambda_1\,{\cup}\,\Lambda_0$.
The polynomials $\Upsilon(\lambda;\mu,a_2)$ generate an ideal $I_{\Upsilon}\,{\subset}\,\Complex[\lambda;\mu,a_2]$.
Gr\"{o}bner basis of $I_{\Upsilon}\,{\cap}\,\Complex[\lambda]$ is $\Delta(\lambda)$.

If $\delta(\mu,a_2)\,{=}\,4 \big(\mu _4-\frac{4}{3}a_2^2\big)^3+27
   \big(\mu_6-\frac{2}{3}a_2 \mu _4+\frac{16 }{27}a_2^3\big)^2$ vanishes then the polynomial 
   $x^3 + 2a_2 x^2 + \mu_4 x + \mu_6$, cf.\,\eqref{CurveHG2Deg}, has a double root. This means
   the curve \eqref{CurveHG2Deg}
   has two double points and its genus is $0$, equivalently $\lambda\,{\in}\,\Lambda_0$. 
The polynomials $\Upsilon(\lambda;\mu,a_2)$ and $\delta(\mu,a_2)$ generate an 
ideal $I_{(\Upsilon,\,\delta)}\,{\subset}\,\Complex[\lambda;\mu,a_2]$.
Gr\"{o}bner basis of 
   $I_{(\Upsilon,\,\delta)}\,{\cap}\,\Complex[\lambda]$ is $\Gamma(\lambda)$.
   
To calculate Gr\"{o}bner bases we use B. Buchberger's method with lexicographic monomial order.
\end{proof}

\begin{remark}
The polynomial $\Delta(\lambda)$  is in fact the discriminant of 
$x^5 \,{+}\, \lambda_4 x^3 + \lambda_6 x^2 + \lambda_8 x + \lambda_{10}$, 
cf.\,\eqref{CurveHG2}, while the polynomial $\delta(\mu,a_2)$ is the discriminant of 
$x^3 + 2a_2 x^2 + \mu_4 x + \mu_6$, cf. \eqref{CurveHG2Deg}.
\end{remark}

Introduce variables $\gamma_4$, $\gamma_6$ by the formulas $\gamma_4 = \mu _4-\tfrac{4}{3}a_2^2$ and
$\gamma_6 = \mu_6-\tfrac{2}{3}a_2 \mu _4+\tfrac{16 }{27}a_2^3$.
Then the above polynomial $\delta(\mu,a_2)$ takes the form $\delta(\gamma)\,{=}\,4\gamma_4^3\,{+}\,27\gamma_6^2$. 
In what follows we shall need the following expressions
\begin{gather}
 \mu_4 =  \gamma_4 + \tfrac{4}{3} a_2^2,\qquad
 \mu_6 = \gamma_6 + \tfrac{2}{3}a_2 \gamma_4 + \tfrac{8}{27}a_2^3. \label{muSubs}
\end{gather}

Equations $\Upsilon(\lambda;\gamma,a_2)\,{=}\,0$ with respect to $(\gamma,\,a_2)$, 
here $\mu$ in \eqref{UpsilonDef} are replaced by $\gamma$
according to \eqref{muSubs}, have no solution 
when $\lambda\,{\in}\,\Lambda_2$, a unique solution for $(\gamma,\,a_2)$
when $\lambda\,{\in}\,\Lambda_1$, and two solutions when $\lambda\,{\in}\,\Lambda_0$.
Indeed, if $\Delta(\lambda)\,{\neq}\,0$ the equations are incompatible. 
Let $\Delta(\lambda)\,{=}\,0$, suppose there exist 
two distinct points $(\gamma,\,a_2)$ and $(\beta,\,b_2)$ corresponding to the same point $\lambda$.
Subtracting $\Upsilon(\lambda;\beta,b_2)\,{=}\,0$ from $\Upsilon(\lambda;\gamma,a_2)\,{=}\,0$ 
then eliminating $\gamma_4\,{-}\,\beta_4$ and $\gamma_6\,{-}\,\beta_6$ we come to a pair
of algebraic equations of order five and four with respect to $t\,{=}\,a_2\,{-}\,b_2$. These equations
have a single common root $t\,{=}0$ iff $\delta(\gamma)\,{\neq}\,0$, thus the points
$(\gamma,\,a_2)$ and $(\beta,\,b_2)$ coincide. Now suppose both $\Delta(\lambda)$ and
$\delta(\gamma)$ vanish, then $(\gamma_6,\,\gamma_4)\,{=}\,(2t^3,{-}3t^2)$ for some value of 
$t\,{\in}\,\Complex$. The system $\Upsilon(\lambda;(2t^3,{-}3t^2),a_2)\,{-}\,
\Upsilon(\lambda;(2s^3,{-}3s^2),b_2)\,{=}\,0$ is satisfied by two solutions: $(s,b_2)\,{=}\,(t,a_2)$
and $(s,b_2)\,{=}\,\big(\frac{2}{3}t+\tfrac{5}{9}a_2,\,t\,{-}\,\frac{2}{3}a_2\big)$.

\section{Frames in strata}\label{s:Frames}
To define a frame in the stratum $\Lambda_2$ we  use a theorem due to V.M.\;Zakalyukin \cite{Z1976},
see also \cite{G1980},  which 
puts into correspondence a vector field $L$ tangent to hypersurface $\Delta(\lambda)\,{=}\,0$ and a
polynomial $p(x,y)$, namely, 
\begin{equation*}
 L f(x,y) = p(x,y) f(x,y) \!\!\mod (\partial_x f,\,\partial_y f),
\end{equation*}
where $f(x,y)=0$ is a curve equation. In our case 
\begin{equation*}
 \Complex[x,y]/(\partial_x f,\,\partial_y f) = \spanOp_{\Complex} \big(1,\,x,\,x^2,\,x^3\big),
\end{equation*}
and four vector fields $\{\ell_0,\,\ell_2,\,\ell_4,\,\ell_6\}$ correspondent to the polynomials 
\begin{align*}
  &p_0(x,y) = 10,&
  &p_2(x,y) = 10 x,&\\
  &p_4(x,y) = 10 x^2 + 6\lambda_4,&
  &p_6(x,y) = 10 x^3 + 6\lambda_4 x + 4\lambda_6&
\end{align*}
provide a basis in $\Lambda_2$. Explicitly 
\begin{gather}
\begin{pmatrix}
 \ell_0\\ \ell_2\\ \ell_4\\ \ell_6
\end{pmatrix}
= V(\lambda)  \begin{pmatrix}
 \partial_{\lambda_4} \\ \partial_{\lambda_6} \\ \partial_{\lambda_8} \\
 \partial_{\lambda_{10}}
\end{pmatrix},\label{LC25}\\
\intertext{where}
V(\lambda) = \begin{pmatrix}
  4 \lambda_4 & 6 \lambda_6 & 8 \lambda_8 & 10 \lambda_{10} \\
 6 \lambda_6 & 8 \lambda_8-\frac{12}{5}\lambda_4^2 & 10
   \lambda _{10}-\frac{8}{5}\lambda_6 \lambda_4  &
   -\frac{4}{5} \lambda_8 \lambda_4  \\
 8 \lambda_8 & 10 \lambda_{10}-\frac{8}{5}\lambda_6\lambda_4  & 
 4 \lambda_8 \lambda_4 -\frac{12}{5}\lambda_6^2 & 
 6 \lambda_{10}\lambda_4 -\frac{6}{5} \lambda_8 \lambda_6 \\
 10 \lambda _{10} & -\frac{4}{5} \lambda_8 \lambda_4 & 
  6 \lambda_{10} \lambda_4 - \frac{6}{5} \lambda_8 \lambda_6  & 
  4 \lambda_{10} \lambda_6 -\frac{8}{5}\lambda_8^2 
 \end{pmatrix}.
\end{gather}

Vector fields $\ell\,{=}\,(\ell_0,\,\ell_2,\,\ell_4,\,\ell_6)$ are
tangent to discriminant variety $\{\lambda\,{\mid}\,\Delta(\lambda)=0\}
\cong \Lambda_1\cup\Lambda_0$, in fact,
\begin{gather}
 \ell_k \Delta(\lambda) = \phi_k \Delta(\lambda),\qquad \phi_k\in\Complex[\lambda],\quad k=0,\,2,\,4,\,6;
 \label{LtangentD}\\
 \phi = (40,\, 0,\, 12 \lambda_4,\,  4\lambda_6). \notag
\end{gather}
Vector fields $\ell$ are tangent to the variety $\{\lambda\,{\mid}\,\Gamma(\lambda)\,{=}\,0\}
\,{\cong}\,\Lambda_0$,
namely
\begin{gather}
 \ell_k \Gamma(\lambda) = \psi_k \Gamma(\lambda),\qquad 
 \psi_k\in \Mat(4; \Complex[\lambda]),\quad k=0,\,2,\,4,\,6; \label{LtangentDD}\\
 \psi_0 = \diag (16,\,18,\,20,\,24),\qquad
 \psi_2 = \begin{pmatrix}
           0 & -6 & 0 & 0 \\
 -\frac{116}{5}\lambda_4 & 0 & \frac{16}{5} & 0 \\
 27 \lambda_6 & -77 \lambda_4 & 0 & 0 \\
 72 \lambda_6 \lambda_4  & 240 \lambda_8-56 \lambda_4^2 & 0 & 0
          \end{pmatrix},\notag \\
  \psi_4 =  \begin{pmatrix}
 -\frac{32}{5}\lambda_4 & 0 & \frac{4}{5} & 0 \\
 \frac{33}{5}\lambda_6 & 5 \lambda_4 & 0 & 0 \\
 24 \lambda_8-\frac{432}{5}\lambda_4^2 & 0 & 12 \lambda_4 & -\frac{12}{5} \\
 144 \lambda_8 \lambda_4+108 \lambda_6^2 -\frac{176}{5}\lambda_4^3 & 0 & 0 & 
 \frac{44}{5}\lambda_4 
            \end{pmatrix},\notag \\
  \psi_6 =  \begin{pmatrix}
  -\frac{7}{5}\lambda _6 & -\frac{7}{5}\lambda _4 & 0 & 0  \\
 4 \lambda _8-\frac{128}{25}\lambda_4^2 & 0 & \frac{16}{25}\lambda _4 & 0 \\
 100 \lambda _{10}-\frac{81}{5}\lambda_6 \lambda_4 &
   -6 \lambda_8 - \frac{81}{5} \lambda_4^2 & 0 & 0 \\
 72 \lambda_8 \lambda_6 -\frac{48}{5} \lambda_6 \lambda _4^2  & 
 40 \lambda_8 \lambda_4 -\frac{48}{5} \lambda_4^3 & 0 & 0  
            \end{pmatrix}. \notag
\end{gather}

It follows from $\det V(\lambda)\,{=}\,\tfrac{16}{5}\Delta(\lambda)$ 
that $\ell$ defines a frame in the stratum $\Lambda_2$; next \eqref{LtangentD} and \eqref{LtangentDD}
imply that restrictions of $\ell$ to the strata $\Lambda_1$ and $\Lambda_0$ 
provide frames on the both strata.
To analyze the restrictions in more detail we need parameterization of $\Lambda_1$ and $\Lambda_0$.
By combining \eqref{CurveHG2Deg} with \eqref{muSubs} and
comparing with \eqref{CurveHG2} we observe that the subset of curves \eqref{CurveHG2} 
with one double point is parameterized as follows
\begin{equation}\label{CoefC25degSubs}
\begin{split}
&\lambda_4 = \gamma_4 - \tfrac{5}{3} a_2^2,\\
&\lambda_6 = \gamma_6 - \tfrac{4}{3}a_2 \gamma_4 - \tfrac{10}{27} a_2^3, \\
&\lambda_8 = -2a_2 \gamma_6  - \tfrac{1}{3} a_2^2 \gamma_4 + \tfrac{20}{27} a_2^4,\\
&\lambda_{10} = a_2^2 \gamma_6 + \tfrac{2}{3} a_2^3 \gamma_4 + \tfrac{8}{27}a_2^5.\\
& 4\gamma_4^3+27 \gamma_6^2 \neq 0.
\end{split}
\end{equation}

\begin{lemma}\label{L:VFdegC25}
The restricted vector fields $(\widetilde{\ell}_0,\,\widetilde{\ell}_2,\,\widetilde{\ell}_4)
\,{=}\,(\ell_0,\,\ell_2,\,\ell_4)|_{\Lambda_1}$ form a frame on 
the stratum~$\Lambda_1$. 
In~terms of parameterization \eqref{CoefC25degSubs} they are expressed as follows
\begin{align*}
 \widetilde{\ell}_0 =&\, 2a_2 \partial_{a_2} + 4\gamma_4 \partial_{\gamma_4} + 6\gamma_6 \partial_{\gamma_6},\\
 \widetilde{\ell}_2 =&\, \tfrac{2}{15}\big(6\gamma_4 + 5 a_2^2\big) \partial_{a_2} 
 + \tfrac{2}{3}\big(9\gamma_6 - 8 a_2 \gamma_4 \big)\partial_{\gamma_4} \,{-}\,
 \tfrac{4}{3}\big(\gamma_4^2 + 6 a_2\gamma_6 \big) \partial_{\gamma_6},\\
 \widetilde{\ell}_4 =&\, \tfrac{2}{45} 
 \big(27\gamma_6 + 9 a_2 \gamma_4 - 40 a_2^3\big) 
 \partial_{a_2} - \tfrac{4}{3}a_2 \big(9\gamma_6 + a_2 \gamma_4\big) \partial_{\gamma_4} 
 - \tfrac{2}{3}a_2 \big(3a_2 \gamma_6 - 4 \gamma_4^2 \big) \partial_{\gamma_6}.
\end{align*}
On $\Lambda_1$ the vector field $\ell_6$ is decomposed into
\begin{align}\label{L6C25deg}
 \ell_6|_{\Lambda_1}  =& {-}a_2^3 \widetilde{\ell}_0
 -a_2^2 \widetilde{\ell}_2 - a_2\widetilde{\ell}_4. 
\end{align}

\end{lemma}
\begin{proof}
 The proof is straightforward.
\end{proof}

\begin{remark}\label{R:Ldefs}
The vector fields $(\widetilde{\ell}_0,\,\widetilde{\ell}_2,\,\widetilde{\ell}_4)$ 
on a curve \eqref{CurveHG2} 
with a double point at $(a_2,0)$ can be expressed in terms of
the three vector fields: $\partial_{a_2}$, $L_0\,{=}\,4\gamma_4 \partial_{\gamma_4} \,{+}\,
6\gamma_6 \partial_{\gamma_6}$, and $L_2\,{=}\,6\gamma_6 \partial_{\gamma_4} \,{-}\,
\tfrac{4}{3} \gamma_4^2 \partial_{\gamma_6}$ as follows
\begin{equation}\label{VFdegC25}
\begin{split}
 \widetilde{\ell}_0 =& 2a_2 \partial_{a_2}  + L_0,\\
 \widetilde{\ell}_2 =&  \tfrac{2}{15}\big(6\gamma_4 + 5 a_2^2\big) \partial_{a_2} 
 - \tfrac{4}{3} a_2 L_0 + L_2 ,\\
 \widetilde{\ell}_4 =& \tfrac{2}{45} 
 \big(27\gamma_6 + 9 a_2 \gamma_4 - 40 a_2^3\big)
 \partial_{a_2} - \tfrac{1}{3}a_2^2 L_0 - 2 a_2 L_2 .
 \end{split}
\end{equation}
The fields $L_0$, $L_2$ are tangent to the variety 
$\{\gamma\,{\mid}\, \delta(\gamma)\,{=}\,0\}$.
\end{remark}

In a similar way, from the generic form of a curve \eqref{CurveHG2} with two double 
points at $(a_2,0)$ and $(b_2,0)$
\begin{equation}\label{CurveHG2DDeg}
- y^2 + (x-a_2)^2 (x-b_2)^2 (x+2a_2+2b_2)  = 0
\end{equation}
we obtain a parameterization of $\Lambda_0$
\begin{equation}\label{CoefC25ddegSubs}
\begin{split}
 &\lambda_4 = -3a_2^2 - 4a_2 b_2 - 3b_2^2,\\
 &\lambda_6 = 2(a_2+b_2)\big(a_2^2+3a_2 b_2+b_2^2\big), \\
 &\lambda_8 = -a_2 b_2 \big(4a_2^2+ 7 b_2 a_2 + 4b_2^2\big),\\
 &\lambda_{10} = 2 a_2^2 b_2^2 (a_2+b_2).
\end{split}
\end{equation}
\begin{lemma}
The restricted vector fields $(\widetilde{\ell}_0,\,\widetilde{\ell}_2)
\,{=}\,(\ell_0,\,\ell_2)|_{\Lambda_0}$ form a frame on 
the stratum~$\Lambda_0$. 
In~terms of parameterization \eqref{CoefC25ddegSubs} they are expressed as follows
\begin{align*}
 \widetilde{\ell}_0 =&\, 2a_2 \partial_{a_2} + 2b_2 \partial_{b_2},\\
 \widetilde{\ell}_2 =&  -\tfrac{2}{5}\big(a_2^2 + 8a_2 b_2 + 6 b_2^2\big) \partial_{a_2} 
 -\tfrac{2}{5}\big(6a_2^2 + 8a_2 b_2 + b_2^2\big) \partial_{b_2}.
\end{align*}
On $\Lambda_0$ the vector fields $\ell_4$ and $\ell_6$ are decomposed into
\begin{subequations}\label{L46C25DDeg}
 \begin{align}
 \ell_4|_{\Lambda_0} =&\, {-} (a_2^2 + a_2 b_2 + b_2^2)\widetilde{\ell}_0 
 - (a_2 + b_2)\widetilde{\ell}_2\\
 \ell_6|_{\Lambda_0} =&\, a_2 b_2 (a_2+b_2)\widetilde{\ell}_0
 + a_2 b_2 \widetilde{\ell}_2. 
\end{align}
\end{subequations}
\end{lemma}
\begin{proof}
 The proof is straightforward.
\end{proof}

\section{Annihilators of sigma-function}\label{s:HeatEqs}
Following \cite{BL2004}, we write down the operators producing
heat equations in a non-holonomic frame in the case of genus $2$ curve \eqref{CurveHG2}:
\begin{subequations}
\begin{align*}
 q_0 =& -u_1 \partial_{u_1} - 3u_3\partial_{u_3} + 3 + \ell_0,\\
 q_2 =& -\tfrac{1}{2} \partial_{u_1 u_1} + \tfrac{4}{5}\lambda_4 u_3 \partial_{u_1}
 - u_1 \partial_{u_3} + \tfrac{3}{10}\lambda_4 u_1^2 
 - \tfrac{1}{10} \big(15\lambda_8 - 4\lambda_4^2\big)u_3^2 + \ell_2,\\
 q_4 =& -\partial_{u_1 u_3} + \tfrac{6}{5}\lambda_6 u_3 \partial_{u_1}
 - \lambda_4 u_3 \partial_{u_3} + \tfrac{1}{5}\lambda_6 u_1^2 - \lambda_8 u_1 u_3 
 \\ &- \tfrac{1}{10}\big(30\lambda_{10}-6\lambda_6 \lambda_4\big)u_3^2 +\lambda_4 + \ell_4,\notag\\
 q_6 =& -\tfrac{1}{2} \partial_{u_3 u_3} + \tfrac{3}{5}\lambda_8 u_3 \partial_{u_1} 
 + \tfrac{1}{10}\lambda_8 u_1^2 - 2\lambda_{10} u_1 u_3 + \tfrac{3}{10}\lambda_8 \lambda_4 u_3^2 
 \\ &  + \tfrac{1}{2}\lambda_6 + \ell_6. \notag
\end{align*}
\end{subequations}
We define sigma-function $\bm{\sigma}(u_3,\,u_1;\,\lambda)$ on genus 2 curve \eqref{CurveHG2} as 
a solution of the equations
\begin{equation*}
 q_k \bm{\sigma}(u_3,\,u_1;\,\lambda) = 0,\quad k=0,\,2,\,4,\,6
\end{equation*}
with the initial condition $\bm{\sigma}(u_3,\,0;\,0) \,{=}\, u_3 $. 
Since the solution is unique \cite{BL2004}, this completely defines the sigma-function.

According to relation \eqref{L6C25deg} 
from Lemma~\ref{L:VFdegC25} the operator $Q_6=-2(q_6+a_2 q_4+a_2^2 q_2
+a_2^3 q_0)|_{\Lambda_1}$ does not include derivatives over $\gamma$ and $a_2$,
namely:
\begin{equation}\label{Q6Def}
 Q_6 =  \Big(\partial_{u_3} + a_2 \partial_{u_1} + 
 a_2^2 u_1 + (\gamma_4 + \tfrac{4}{3} a_2^2) a_2 u_3\Big)^2 
 -\big(\gamma_6 + \tfrac{5}{3}a_2 \gamma_4 + \tfrac{125}{27}a_2^3\big).
\end{equation}
Introduce a new variable $U_1$ by the formula $u_1\,{=}\,U_1\,{+}\,a_2 u_3$, then
\eqref{Q6Def} becomes an ordinary differential operator
\begin{equation}\label{Q6SMpl}
 Q_6 =  D^2 - d(a_2,\gamma)^2
\end{equation}
with the operator $D$ and the function $d(a_2,\gamma)$ given by
\begin{align*}
 &D = \partial_{u_3} + 
 a_2^2 U_1 + (\gamma_4 + \tfrac{7}{3} a_2^2) a_2 u_3
 &d(a_2,\gamma)^2 = \gamma_6 + \tfrac{5}{3}a_2 \gamma_4 + (\tfrac{5}{3}a_2)^3.
\end{align*}
The operator 
$Q_4\,{=}\,{-}(q_4\,{+}\,2a_2 q_2\,{+}\,3a_2^2 q_0)|_{\Lambda_1}$ has the form
\begin{multline}\label{Q4Def}
 Q_4 =  \Big(\partial_{U_1} + 
 2 a_2 U_1 + (\gamma_4 + \tfrac{28}{3} a_2^2) u_3\Big) D
 -\tfrac{6}{5} d(a_2,\gamma) \partial_{a_2}\big(d(a_2,\gamma)\, \cdot\,\big)
 \\ - \tfrac{1}{5} \big(U_1^2 + 12 a_2 U_1 u_3 + 3(\gamma_4 + 7a_2^2) u_3^2 \big) d(a_2,\gamma)^2.
\end{multline}

Then $Q_0\,{=}\,q_0|_{\Lambda_1}$ and $Q_2\,{=}\,q_2 \,{+}\, \tfrac{4}{3} a_2 q_0|_{\Lambda_1}$
take the form
\begin{subequations}\label{QopsC25deg}
\begin{align}
 Q_0 =& -U_1 \partial_{U_1} - 3u_3\partial_{u_3} + 2 a_2 \partial_{a_2} + L_0 + 3,\\
 Q_2 =& -\tfrac{1}{2} \partial_{U_1 U_1} 
 - \tfrac{1}{3}  a_2 \big(U_1 + 3 a_2 u_3\big) \partial_{U_1}
 - \big(U_1 + 5 a_2 u_3 \big) \partial_{u_3} 
 \\ & +  \tfrac{2}{15}\big(6\gamma_4 + 25 a_2^2\big) \partial_{a_2}
 + L_2 + \tfrac{1}{10}\big(3\gamma_4 - 5a_2^2\big) \big(U_1 + 2 a_2 u_3\big) U_1 \notag \\ & + 
  \tfrac{1}{30}\big(90 a_2 \gamma_6  + 12\gamma_4^2 - 16a_2^2 \gamma_4 
 - 15 a_2^4\big) u_3^2 + 4a_2. \notag
\end{align}
\end{subequations}

A solution $\mathcal{Z}(u_3,U_1,a_2, \gamma)$ of the system
$$Q_k \mathcal{Z} = 0,\quad k=0,\,2,\,4,\,6,\qquad
\mathcal{Z}(u_3,0,0,0) = u_3$$
at $U_1\,{=}\,u_1\,{-}\,a_2 u_3$ is a degenerate 
sigma-function and coincides with $\bm{\sigma}(u_3,\,u_1;\,\lambda)$ restricted to $\Lambda_1$.
We construct this solution explicitly in the next section.

\section{Degenerate sigma-function}\label{s:DegenSigma}

\begin{theorem}\label{Th:1}
Suppose $\lambda\in \Lambda_1$. 
Sigma-function associated with a curve \eqref{CurveHG2} has the form
\begin{gather}\label{SolG2}
\begin{split}
\bm{\sigma}(u_3,u_1,\lambda)&|_{\Lambda_1} =
\frac{\rme^{-\tfrac{3}{5} \wp(\alpha)  \Big( 
\big(\tfrac{1}{2}\gamma_4 + \tfrac{3}{25} \wp(\alpha)^2\big)u_3^2 + \tfrac{2}{5} \wp(\alpha) u_1 u_3 
+ \tfrac{1}{6} u_1^2 \Big)}}{ \wp'(\alpha) \sigma (\alpha)} \times \\ &\times
\bigg(\sigma \big(\alpha + u_1 - \tfrac{3}{5}\wp(\alpha) u_3\big)
\rme^{\tfrac{1}{2} \wp'(\alpha) u_3  - \zeta (\alpha)\big(u_1 - \tfrac{3}{5}\wp(\alpha) u_3\big)}
\\ &\quad - \sigma \big(\alpha - u_1 + \tfrac{3}{5}\wp(\alpha) u_3\big) 
\rme^{-\tfrac{1}{2} \wp'(\alpha) u_3  + \zeta (\alpha)\big(u_1 - \tfrac{3}{5}\wp(\alpha) u_3\big)} \bigg),
\end{split}
\end{gather}
where  $\sigma$, $\zeta$, $\wp$ are Weierstrass functions associated with the curve 
\eqref{CurveHG1}, and $\alpha$ is defined by $\wp(\alpha) = \tfrac{5}{3} a_2$.
\end{theorem}
\begin{proof}
First, we consider the equation
\begin{equation*}
 Q_6 \mathcal{Z}(u_3,U_1,a_2,\gamma) = 0
\end{equation*}
where $Q_6$ is defined by \eqref{Q6SMpl}. The gauge transformation
\begin{equation}\label{Zansatz}
\mathcal{Z}(u_3,U_1,a_2,\gamma) =
\exp\Big\{{-}\tfrac{1}{2}a_2\big(\gamma_4 +\tfrac{7}{3}a_2^2\big)u_3^2 - a_2^2 U_1 u_3 \Big\}
\rho(u_3,U_1,a_2,\gamma)
\end{equation}
leads to a simpler equation
$$\partial_{u_3 u_3} \rho(u_3,U_1,a_2,\gamma) - d(a_2,\gamma)^2 \rho(u_3,U_1,a_2,\gamma) = 0.$$
As fundamental solutions of the equation we choose 
$c_{\epsilon}(U_1,a_2,\gamma)\exp (\epsilon d(a_2,\gamma)u_3)$,
where $\epsilon$ is unary operator: $\epsilon\,{=}\,\pm $. Then
\begin{equation}\label{rho}
 \rho(u_3,U_1,a_2,\gamma) = {c}_{+}(U_1,a_2,\gamma) \rme^{ u_3 d(a_2,\gamma)}
+ {c}_{-}(U_1,a_2,\gamma) \rme^{- u_3 d(a_2,\gamma)}.
\end{equation}

Next, consider the equation
\begin{equation*}
 Q_4 \mathcal{Z}(u_3,U_1,a_2,\gamma) = 0.
\end{equation*}
Taking into account \eqref{Zansatz} and \eqref{rho} we obtain the following equations for $c_{\epsilon}$
\begin{equation*}
 \epsilon \partial_{U_1} c_{\epsilon}   - \tfrac{6}{5} \partial_{a_2} \big(d(a_2,\gamma) c_{\epsilon} \big)
 = \Big({-} 2\epsilon a_2  U_1 + \tfrac{1}{5} d(a_2,\gamma) U_1^2 \Big) c_{\epsilon}. 
\end{equation*}
The substitution 
$$c_{\epsilon}(U_1,a_2,\gamma)\,{=}\,\exp\big\{\varphi_{\epsilon}(U_1,a_2,\gamma)\big\}/d(a_2,\gamma)$$
leads to a linear non-homogeneous PDE
\begin{equation}\label{Q4eqG2}
\big(\epsilon \partial_{U_1}  - \tfrac{6}{5} d(a_2,\gamma) \partial_{a_2}\big) \varphi_{\epsilon} 
 = - 2 \epsilon a_2 U_1 + \tfrac{1}{5} d(a_2,\gamma) U_1^2.
\end{equation}
We solve the associated homogeneous equation by the method of characteristics:
$$-\epsilon \dd U_1 = \frac{5}{6} \frac{\dd a_2}{d(a_2,\gamma)}
= \frac{\dd \big(\frac{5}{3}a_2\big)}
{-2\sqrt{\big(\tfrac{5}{3}a_2\big)^3+\tfrac{5}{3}a_2 \gamma_4 +\gamma_6}}.$$
The characteristics is defined by the equation
\begin{equation*}
 \alpha (a_2,\gamma) + \epsilon U_1  = \const,
\end{equation*}
where
\begin{equation*}
 \alpha (a_2,\gamma) =  \int_\infty^{\frac{5}{3}a_2}\frac{\dd X}{- 2\sqrt{X^3 + \gamma_4 X + \gamma_6}},
 \qquad \deg \alpha = 1.
\end{equation*}
We write down a general solution of the homogeneous equation as 
$$\varphi^{(h)}_{\epsilon} = \log s_{\epsilon} \big(\alpha(a_2,\gamma) + \epsilon U_1,\gamma \big).$$

In what follows we need elliptic functions $\sigma$, $\zeta$, $\wp$, $\wp'$  associated with the curve
\eqref{CurveHG1}. Here
$$\wp(\alpha) = \tfrac{5}{3} a_2\qquad
 \wp'(\alpha) = 2 d(a_2,\gamma) = 
 - 2 \sqrt{\big(\tfrac{5}{3}a_2\big)^3+\tfrac{5}{3}a_2 \gamma_4 +\gamma_6}. $$
The functions $\wp$, $\wp'$ satisfy the equation
$(\wp')^2 = 4\wp^3 + 4\lambda_4 \wp + 4\lambda_6$, thus, they are standard Weierstrass 
functions with the invariants $(g_2,g_3)\,{=}\,({-}4\lambda_4,\,{-}4\lambda_6)$, see \cite{BatErd}.

Next, we construct a particular solution of non-homogeneous equation \eqref{Q4eqG2} in the form
$$\varphi_{\epsilon}^{\text{(nh)}} = C_2(a_2,\gamma)U_1^2 + C_1(a_2,\gamma) U_1 + C_0(a_2,\gamma).$$
By substituting the ansatz and collecting coefficients at the powers of $U_1$ we obtain a system of equations
for $C_2$, $C_1$ and $C_0$:
\begin{gather*}
 \partial_{a_2 }C_2 = -\tfrac{1}{6},\quad 
 \tfrac{6}{5} d(a_2,\gamma) \partial_{a_2} C_1= 2 \epsilon \big(C_2 + a_2\big),\quad
 \tfrac{6}{5} d(a_2,\gamma) \partial_{a_2 } C_0 = \epsilon C_1.
\end{gather*}
Observe that $\partial_{\alpha}\,{=}\,\tfrac{6}{5}d(a_2,\gamma)\partial_{a_2}$.
Whence
\begin{align*}
 &C_2(a_2,\gamma) = {-}\tfrac{1}{6} a_2, &\\
 &\partial_{\alpha} C_1 = \epsilon \tfrac{5}{3}a_2 = \epsilon \wp(\alpha)& 
 &\Rightarrow&
 &C_1(a_2,\gamma) = - \epsilon \zeta(\alpha),&\\
 &\partial_{\alpha} C_0 = - \zeta(\alpha)& &\Rightarrow&
 &C_0(a_2,\gamma) = -\log \sigma(\alpha),&
\end{align*}
above we have used the standard relations:
\begin{gather*}
\zeta(u) = -\int_{\infty}^{u} \wp(v) \dd v,\qquad
\log \sigma(u) = \int_{\infty}^{u} \zeta(v) \dd v.
\end{gather*}
Summing up, the general solution $\varphi^{(h)}_{\epsilon}  \,{+}\, \varphi^{(nh)}_{\epsilon}$ 
of \eqref{Q4eqG2} has the form
\begin{gather*}
 \varphi_\epsilon (U_1,a_2,\gamma) = \log s_{\epsilon}\big(\alpha(a_2,\gamma) + \epsilon U_1,\gamma\big)
- \tfrac{1}{6}a_2 U_1^2 - \epsilon  \zeta (\alpha) U_1 - 
\log \sigma(\alpha).
\end{gather*}
Therefore, we come to the following expression for $c_{\epsilon}$:
\begin{equation}\label{cForm}
c_{\epsilon} (U_1,a_2,\gamma)
= \frac{s_{\epsilon}\big(\alpha(a_2,\gamma) + \epsilon U_1,\gamma\big)}
{\wp'(\alpha(a_2,\gamma)) \sigma(\alpha(a_2,\gamma))} 
\rme^{-\tfrac{1}{6}a_2 U_1^2 - \epsilon \zeta(\alpha(a_2,\gamma))U_1}.
\end{equation}
Taking into account the form \eqref{cForm} of dependence of 
$c_{\epsilon}$ on $a_2$, for the next step we change the variables on $\Lambda_1$ from
$(a_2,\,\gamma_4,\,\gamma_6)$ to $(\alpha,\,\gamma_4,\,\gamma_6)$:
\begin{equation*}
c_{\epsilon} (U_1,\alpha,\gamma)
= \frac{s_{\epsilon}\big(\alpha + \epsilon U_1,\gamma\big)}
{\wp'(\alpha) \sigma(\alpha)} 
\rme^{-\tfrac{1}{10} \wp(\alpha) U_1^2 - \epsilon \zeta(\alpha)U_1}.
\end{equation*}
Under the change of variables the operators $Q_2$, $Q_0$ map to new operators $\widetilde{Q}_2$, 
$\widetilde{Q}_0$, where the map is defined by the following formula (cf. Remark~\ref{R:Ldefs})
\begin{align*}
 &\big(\partial_{a_2},\,L_2,\,L_0\big) \mapsto \bigg(\frac{5}{3\wp'(\alpha)} \partial_{\alpha},\,
  L_2 - \frac{L_2 \big(\wp(\alpha)\big)}{\wp'(\alpha)} \partial_{\alpha},\,
  L_0 - \frac{L_0 \big(\wp(\alpha)\big)}{\wp'(\alpha)} \partial_{\alpha} \bigg).
\end{align*}

Applying the operator $\widetilde{Q}_2$ to $\mathcal{Z}(u_3,U_1,\tfrac{3}{5}\wp(\alpha),\gamma)$ 
with ansatz \eqref{cForm} and using the relations
\begin{align*}
 &L_2 \sigma(\alpha) = \sigma(\alpha) \big({-}\tfrac{1}{6}\gamma_4 \alpha^2 + 
 \tfrac{1}{2} \zeta(\alpha)^2 - \tfrac{1}{2}\wp(\alpha)\big),\\
 &L_2 \zeta(\alpha) = - \tfrac{1}{3}\gamma_4 \alpha 
 - \zeta(\alpha) \wp(\alpha) - \tfrac{1}{2}\wp'(\alpha),\\
 &L_2 \wp(\alpha) = \tfrac{4}{3}\gamma_4 +2\wp(\alpha)^2 + \zeta(\alpha) \wp'(\alpha),\\
 &L_2 \wp'(\alpha) =  \zeta(\alpha) 
 \big(6\wp(\alpha)^2 + 2\gamma_4\big) + 3 \wp(\alpha) \wp'(\alpha).
\end{align*}
we come to the equation
\begin{equation}\label{Q2sigmaG1}
 \Big({-}\tfrac{1}{2} \partial_{U_1 U_1} + 
 \tfrac{1}{6} \gamma_4 \big(\alpha + \epsilon U_1 \big)^2 
 + L_2 \Big) s_{\epsilon}(\alpha + \epsilon U_1,\gamma) = 0.
\end{equation}
Similarly, the operator $\widetilde{Q}_0$ leads to 
the equation
\begin{equation}\label{Q0sigmaG1}
 \Big({-}(\alpha + \epsilon U_1) \partial_{U_1} 
 + L_0 + 1\Big) s_{\epsilon}(\alpha + \epsilon U_1,\gamma) = 0.
\end{equation}

Further, consider the power series expansion for $\mathcal{Z}(u_3,0,\tfrac{3}{5}\wp(\alpha),\gamma)$
in $u_3$ near zero. We obtain
\begin{equation*}
 \mathcal{Z}(u_3,0,\tfrac{3}{5}\wp(\alpha),\gamma) =
 \frac{s_{+}(\alpha,\gamma)
 +s_{-}(\alpha,\gamma)}{\sigma(\alpha) \wp'(\alpha)}
 + \frac{s_{+}(\alpha,\gamma)
 - s_{-}(\alpha,\gamma)}{2\sigma(\alpha)} u_3 + O(u_3^2). 
\end{equation*}
Comparing the expansion with the  initial condition $\mathcal{Z}(u_3,0,0,0)\,{=}\,u_3$ 
for entire function $\mathcal{Z}$ and taking into account that at $\gamma\,{=}\,(0,0)$ 
the value of $\alpha(a_2,\gamma)$ tends to infinity as 
$a_2\,{\to}\,0$ we find 
\begin{align*}
 &s_{+}\big(\alpha,0\big) = - s_{-}\big(\alpha,0\big),\\ 
 &s_{+}\big(\alpha,0\big) = \sigma(\alpha)|_{\gamma=0} =\alpha. 
\end{align*}
Therefore, $s_{\epsilon}\big(\alpha,0\big) \,{=}\, \epsilon \alpha$. Thus, the initial condition singles out
a unique solution of equations \eqref{Q2sigmaG1} and \eqref{Q0sigmaG1} that is 
\begin{equation*}
s_{\epsilon}(\alpha + \epsilon U_1,\gamma ) = \epsilon \sigma(\alpha + \epsilon U_1).
\end{equation*}

Combining all of the above results we write down the final expression for $\mathcal{Z}$.
\end{proof}

\begin{remark}\label{R:BakerF}
 Note that the genus 2 degenerate sigma-function \eqref{SolG2} can be represented with the help 
 of elliptic Baker function $\Phi$ 
\begin{gather}\label{BakerF}
 \Phi(u,\alpha) = \frac{\sigma(\alpha-u)}{\sigma(\alpha)\sigma(u)} \,
 \rme^{\zeta(\alpha) u}.
\end{gather}
Indeed, we have
\begin{multline*}
\bm{\sigma}(u_3,u_1,\lambda)|_{\Lambda_1} = \\ -
\rme^{-\tfrac{3}{5} \wp(\alpha)  \Big( 
\big(\tfrac{1}{2}\gamma_4 + \tfrac{3}{25} \wp(\alpha)^2\big)u_3^2 + \tfrac{2}{5} \wp(\alpha) u_1 u_3 
+ \tfrac{1}{6} u_1^2 \Big)} \frac{\sigma \big(u_1 - \tfrac{3}{5}\wp(\alpha) u_3\big)}{ \wp'(\alpha)}  
\times \\ \times 
\bigg(\Phi \big({-}u_1 + \tfrac{3}{5}\wp(\alpha) u_3,\alpha\big)
\rme^{\tfrac{1}{2} \wp'(\alpha) u_3 } + \Phi \big(u_1 - \tfrac{3}{5}\wp(\alpha) u_3,\alpha\big)  
\rme^{-\tfrac{1}{2} \wp'(\alpha) u_3 } \bigg).
\end{multline*}
\end{remark}

\begin{remark}\label{R:SigmaF2}
 Visibly right hand side of \eqref{SolG2} is singular when $\sigma(\alpha)\,{=}\,0$
 or $\wp'(\alpha)\,{=}\,0$. The first case corresponds to $a_2\,{=}\,\infty$ which does not belong to
 $\Lambda_1$, otherwise the equation \eqref{CurveHG2} would not include the term $x^5$. 
 In the second case $2 \alpha$ is a period, say $\omega_1$,
 of Weierstrass functions, that is $\frac{5}{3}a_2$ becomes
 a branch point $e_i$ of \eqref{CurveHG1}. 
 Then 
\begin{gather}\label{SolG2Lemn}
\begin{split}
\bm{\sigma}(u_3,&u_1,\lambda)|_{\Lambda_1} = u_3\,
\rme^{-\tfrac{3}{5} \wp(\alpha)  \Big( 
\big(\tfrac{1}{2}\gamma_4 + \tfrac{3}{25} \wp(\alpha)^2\big)u_3^2 + \tfrac{2}{5} \wp(\alpha) u_1 u_3 
+ \tfrac{1}{6} u_1^2 \Big)} \times \\ &\qquad \qquad\times
 \frac{\sigma \big(\alpha + u_1 - \tfrac{3}{5}\wp(\alpha) u_3 \big)}{\sigma (\alpha)}
\rme^{- \zeta(\alpha) \big(u_1 - \tfrac{3}{5}\wp(\alpha) u_3\big)}\\
&= u_3\, \rme^{-\tfrac{3}{5} e_i \Big( 
\big(\tfrac{1}{2}\gamma_4 + \tfrac{3}{25} e_i^2\big)u_3^2 + \tfrac{2}{5} e_i u_1 u_3 
+ \tfrac{1}{6} u_1^2 \Big)} \sigma_i \big(u_1 - \tfrac{3}{5}e_1 u_3 \big),
\end{split}
\end{gather}
where $\sigma_i(u)=\exp(-u\eta_i)\sigma(u+\omega_i)/\sigma(\omega_i)$ with $e_i=\wp(\omega_i)$, 
$\eta_i=\zeta(\omega_i)$ denotes a sigma-function with characteristic  \cite{BatErd}, p.348 eq.(22).
\end{remark}

\begin{theorem}\label{Th:2}
Suppose $\lambda\in \Lambda_0$. 
Sigma-function associated with a curve \eqref{CurveHG2} has the form
\begin{multline}\label{SolG2d}
\bm{\sigma}(u_3,u_1,\lambda)|_{\Lambda_0} =
\frac{\rme^{\tfrac{1}{2} \big( 
3 a_2 b_2 (a_2 + b_2) u_3^2 + 2 a_2 b_2 u_1 u_3 - (a_2 + b_2) u_1^2\big)}}
{4 (a_2 - b_2)} \times \\ \times
\bigg(\cosh\big(\sqrt{2a_2+3b_2}(u_1-a_2 u_3)\big)
\frac{\sinh\big(\sqrt{3a_2+2b_2}(u_1-b_2 u_3)\big)}{\sqrt{3a_2+2b_2}}
\\ -  \cosh\big(\sqrt{3a_2+2b_2}(u_1-b_2 u_3)\big) 
\frac{\sinh\big(\sqrt{2a_2+3b_2}(u_1-a_2 u_3)\big)}{\sqrt{2a_2+3b_2}}\bigg).
\end{multline}
\end{theorem}
The theorem is proven by an argument similar to the proof of Theorem~\ref{Th:1}.

\section{Applications}\label{s:Appl}
\subsection{A generalized Jacobi inversion problem}\label{ss:GJIP}
Let $(X_1,Y_1)$ and $(X_2,Y_2)$ be a pair of points on the elliptic curve \eqref{CurveHG1}.
Consider an inversion problem for integrals
\begin{equation}\label{IJPrG2r}
\begin{aligned}
 &\int_{\infty}^{(X_1,Y_1)} \frac{\dd X}{-2Y} 
 + \int_{\infty}^{(X_2,Y_2)} \frac{\dd X}{-2Y} = U_1,\\
 &\int_{\infty}^{(X_1,Y_1)} \frac{\dd X}{-2Y(X-A)} + 
 \int_{\infty}^{(X_2,Y_2)} \frac{\dd X}{-2Y(X-A)} = U_3.
\end{aligned}
\end{equation}
Denote 
$$\mathcal{Z} = \bm{\sigma}\big(U_3,U_1+\tfrac{3}{5}A U_3,\lambda(A,\gamma)\big),$$
and $\lambda(A,\gamma)$ is defined by
\begin{equation}\label{ParamC25dX}
\begin{split}
&\lambda_4 = \gamma_4 - \tfrac{3}{5} A^2,\\
&\lambda_6 = \gamma_6 - \tfrac{4}{5} A \gamma_4 - \tfrac{2}{25} A^3, \\
&\lambda_8 = -\tfrac{6}{5} A \gamma_6  - \tfrac{3}{25} A^2 \gamma_4 + \tfrac{12}{125} A^4,\\
&\lambda_{10} = \tfrac{9}{25} A^2 \gamma_6 + \tfrac{18}{125} A^3 \gamma_4 + \tfrac{72}{3125} A^5.
\end{split}
\end{equation}
Further, let 
$$\mathcal{P}_{ij} = -\partial_{U_i U_j} \log \mathcal{Z}\quad \text{and}\quad
\mathcal{P}_{ijk} = -\partial_{U_i U_j U_k} \log \mathcal{Z}.$$

\begin{corollary}\label{C:JacInvPr}
The solution of a generalized Jacobi inversion problem \eqref{IJPrG2r} 
is given by the formulas
\begin{equation}\label{IJPrwpG2d}
\begin{gathered}
X_1 + X_2  = \mathcal{P}_{11} + \tfrac{4}{5}A,\\ 
X_1 X_2 = - \mathcal{P}_{13} + A \mathcal{P}_{11} + \tfrac{4}{25}A^2,\\
Y_k = -\frac{1}{2} \mathcal{P}_{111} 
- \frac{\mathcal{P}_{113}}{2(X_k-A)},\quad k=1,\,2.
\end{gathered}
\end{equation}
\end{corollary}

\begin{proof}
Consider the Jacobi inversion problem on a genus 2 curve of the form \eqref{CurveHG2}
\begin{equation}\label{IJPrG2}
\begin{aligned}
 &\int_{P_0}^{(x_1,y_1)} \frac{\dd x}{-2y} + \int_{P_0}^{(x_2,y_2)}\frac{\dd x}{-2y} = u_3,\\
 &\int_{P_0}^{(x_1,y_1)} \frac{x\,\dd x}{-2y} + \int_{P_0}^{(x_2,y_2)}\frac{x\,\dd x}{-2y} = u_1.
\end{aligned}
\end{equation}
The pair of points $(x_1,y_1)$ and $(x_2,y_2)$ on the curve is defined by formulas 
\begin{equation}\label{IJPrwpG2}
\begin{gathered}
x_1 + x_2 = \wp_{11},\qquad x_1 x_2 = -\wp_{13},\\
y_k = -\tfrac{1}{2} \big(x_k \wp_{111} + \wp_{113}),\quad k=1,\,2,
\end{gathered}
\end{equation}
where $\wp_{ij} = -\partial_{u_i u_j} \log \bm{\sigma}(u_3,u_1,\lambda)$ and
$\wp_{ijk} = -\partial_{u_i u_j u_k} \log \bm{\sigma}(u_3,u_1,\lambda)$.
For more details see \cite{Baker}.

Indeed, relations \eqref{IJPrwpG2} hold for all values of $u$ and $\lambda$ where
sigma-function  does not vanish. 
Consider \eqref{IJPrG2} with parameters $\lambda$ as in \eqref{ParamC25dX}. The
substitution
\begin{equation}\label{TransG12}
\begin{gathered}
 x = X-\tfrac{2}{5}A,\qquad y = Y(X-A),\\
 u_3=U_3,\qquad u_1 = U_1 + \tfrac{3}{5} A U_3
\end{gathered}
\end{equation}
transforms the problem \eqref{IJPrG2} to the problem \eqref{IJPrG2r}.
Consequently, \eqref{IJPrwpG2} transforms to \eqref{IJPrwpG2d}.
\end{proof}

Introducing the following notation
\begin{subequations}\label{SPnot}
\begin{gather}
\mathcal{P} = \frac{\sigma \big(\alpha + U_1 \big)}
{\sigma \big(\alpha - U_1\big)}\, \rme^{ \wp'(\alpha) U_3  - 2\zeta (\alpha)U_1},\label{Snot}\\
\mathcal{S} = \frac{1}{2\big(\wp(U_1)-\wp(\alpha)\big)}
\bigg(\wp'(U_1) - \wp'(\alpha) \frac{\mathcal{P}+1}{\mathcal{P}-1}\bigg), \label{Pnot}
\end{gather}
\end{subequations}
where $\wp(\alpha)\,{=}\,A$, we present explicit expressions for \eqref{IJPrwpG2d}:
\begin{subequations}\label{IJPrwpG2dG}
\begin{gather}
X_1 + X_2  = \mathcal{S}^2 - \wp(U_1),\\ 
X_1 X_2 = \wp(U_1)\mathcal{S}^2 - \wp'(U_1) \mathcal{S} 
- \wp(\alpha)\big(\wp(U_1) + \wp(\alpha)\big)
 + \frac{\wp'(U_1)^2 - \wp'(\alpha)^2}{4\big(\wp(U_1)-\wp(\alpha)\big)},\\
 \intertext{and from $Y_k = {-}\frac{1}{2} 
\Big(X_k\partial_{U_1} (X_1 + X_2) - \partial_{U_1} (X_1 X_2)\Big)/\big(X_k-\wp(\alpha)\big)$}
\begin{split}\label{IJPrwpG2dGY}
Y_k =& -\frac{X_k - \wp(U_1)}{X_k - \wp(\alpha)}\, \mathcal{S}^3 
+ \frac{\wp'(U_1)}{X_k - \wp(\alpha)}\,\mathcal{S}^2  \\
&+ \bigg(2\wp(U_1) + \wp(\alpha) + 
 \frac{\wp'(U_1)^2-\wp'(\alpha)^2}{4\big(X_k - \wp(\alpha)\big)\big(\wp(U_1)-\wp(\alpha)\big)}\bigg) 
 \mathcal{S} + \frac{1}{2}\wp'(U_1). 
\end{split}
\end{gather}
\end{subequations}

\begin{example}
 In the case when $A$ is a branch point, say $e_1\,{=}\,\wp(\omega/2)$, of the curve \eqref{CurveHG1}
 the function $\bm{\sigma}$ is simplified dramatically, cf. \eqref{SolG2Lemn}.
 However formula \eqref{IJPrwpG2dGY} fails for one of the roots. The explicit solution has the form
\begin{gather}\label{IJPrwpG2dE}
(X_1,Y_1)= (e_1,0),\quad
(X_2,Y_2) = \big(\wp(U_1 + \omega/2), -\tfrac{1}{2}\wp'(U_1 + \omega/2)\big).
\end{gather}
\end{example}

Introducing variables $\xi_k$ by the equalities $\wp(\xi_k)\,{=}\,X_k$, $k\,{=}\,1,\,2$ 
we rewrite the problem \eqref{IJPrG2r} in the form
\begin{equation}\label{BAPr0}
\begin{gathered}
 \xi_1 + \xi_2 = U_1,\\
 \int_0^{\xi_1} \frac{\dd\xi}{\wp(\xi)-\wp(\alpha)} +
 \int_0^{\xi_2} \frac{\dd\xi}{\wp(\xi)-\wp(\alpha)} = U_3.
\end{gathered}
\end{equation}
With the help of
\begin{equation*}
 \frac{\wp'(\alpha)}{\wp(\alpha)-\wp(\xi)} = 2\zeta(\alpha) - \zeta(\alpha-\xi) - \zeta(\alpha+\xi)
\end{equation*}
we explicitly integrate and reduce \eqref{BAPr0} to the following system
\begin{equation}\label{BAPr}
 \xi_1 + \xi_2 = U_1,\qquad
 \frac{\sigma (\alpha-\xi_1)\sigma (\alpha-\xi_2)}
 {\sigma (\alpha+\xi_1)\sigma (\alpha+\xi_2)} = \rme^{-2\zeta(\alpha)U_1 + \wp'(\alpha) U_3} .
\end{equation}

\begin{example}
For the rational limit $(\gamma_4,\gamma_6)\,{=}\,0$ 
\begin{equation*}
\sigma(\xi) = \xi,\quad \zeta(\xi) = \xi^{-1},\quad \wp(\xi)=\xi^{-2},\quad
\wp'(\xi) = -2 \xi^{-3}
\end{equation*}
the problem \eqref{IJPrG2r} with usage of \eqref{BAPr} is solved explicitly: 
\begin{equation*}
 \xi_1 + \xi_2 = U_1,\qquad 
 \xi_1 \xi_2  = -\alpha^2 + \alpha U_1 
 \bigg[\tanh\bigg(\frac{U_3}{\alpha^3}+\frac{U_1}{\alpha}\bigg)\bigg]^{-1}.
\end{equation*}
In the rational limit the same relations are obtained from \eqref{IJPrwpG2dG} with
$X_k\,{=}\,\xi_k^{-2}$.
\end{example}

Consider the equation with respect to $\xi$
\begin{equation}\label{EBA}
 e^\varkappa = \frac{\sigma(\xi-\alpha)\sigma(\xi+\beta)}
 {\sigma(\xi+\alpha)\sigma(\xi-2\alpha +\beta)}.
\end{equation}
Similar equations appear in the theory of Bethe ansatz, see \cite{Gaudin} and many other publications. 
By combining the substitutions
\begin{equation*}
 \beta = \alpha-U_1 ,\qquad \varkappa = -2\zeta(\alpha)U_1 + \wp'(\alpha) U_3.
\end{equation*}
the equation \eqref{EBA} is reduced to \eqref{BAPr} 
and has two solutions $\xi_1$, $\xi_2$ defined by \eqref{IJPrwpG2dG}.
\begin{remark}
 The ratio $\displaystyle\frac{\sigma(\xi-\alpha)\sigma(\xi+\beta)}
 {\sigma(\xi+\alpha)\sigma(\xi-2\alpha +\beta)}$ can be represented as a rational function in
 $\wp(\xi)$ and $\wp'(\xi)$. Then the equation \eqref{EBA} is transformed to
 \begin{equation*}
  e^\varkappa =  \frac{\sigma(\beta)}{\sigma(2\alpha-\beta)} \cdot 
  \frac{\wp(\alpha)-\wp(2\alpha-\beta)}{\wp(\alpha)-\wp(\beta)} \cdot
  \frac{\begin{vmatrix}
         \wp'(\xi) & \wp(\xi) & 1 \\
         \wp'(\alpha) & \wp(\alpha) & 1 \\
         -\wp'(\beta) & \wp(\beta) & 1
        \end{vmatrix}}{\begin{vmatrix}
         \wp'(\xi) & \wp(\xi) & 1 \\
         -\wp'(\alpha) & \wp(\alpha) & 1 \\
         \wp'(2\alpha-\beta) & \wp(2\alpha-\beta) & 1
        \end{vmatrix}} .
 \end{equation*}
This is equivalent to an equation of the form
$A\wp'(\xi)+B\wp(\xi)+C=0$, which apparently has three roots. Two of the roots are functions in $\varkappa$
and provide a solution of \eqref{EBA} and in fact are the same as defined by \eqref{IJPrwpG2dG}. 
The extra root $\big(\wp(\xi),\wp'(\xi)\big)\,{=}\,
\big(\wp(\alpha\,{-}\,\beta),{-}\wp'(\alpha-\beta)\big)$ is independent of $\varkappa$.
\end{remark}

\subsection{Schr\"{o}dinger equation with periodic potential}\label{ss:SchEq}
Introduce the function
\begin{multline}\label{BackerFC25}
 \Phi\big((u_3,u_1),(\beta_3,\beta_1)\big)) = \frac{\bm{\sigma}(\beta_3 - u_3,\beta_1 - u_1,\lambda)}
 {\bm{\sigma}(u_3,u_1,\lambda)} \times \\ \times
 \exp\bigg({-}u_3 \int_{\infty}^{(b,y(b))} \frac{(3x^3 + \lambda_4 x)\,\dd x}{-2y} - 
 u_1 \int_{\infty}^{(b,y(b))} \frac{x^2\,\dd x}{-2y}\bigg),
\end{multline}
where $(\beta_1,\beta_3)$ is the image of the point $(b,y(b))$ on the genus 2 curve \eqref{CurveHG2}
under the Abelian map
\begin{equation*}
 \beta_3 = \int_{\infty}^{(b,y(b))} \frac{\dd x}{-2y},\qquad 
 \beta_1 = \int_{\infty}^{(b,y(b))} \frac{x\,\dd x}{-2y}.
\end{equation*}
The 1-forms $\frac{3x^3 + \lambda_4 x}{-2y}\dd x$ and $\frac{x^2}{-2y}\dd x$ 
are second kind differentials associated to the first kind differentials $\frac{1}{-2y}\dd x$ and 
$\frac{x}{-2y}\dd x$.
The function $\Phi\big((u_3,u_1),(\beta_3,\beta_1)\big))$ is a genus 2 analog  
of the elliptic Baker function \eqref{BakerF}. 

Next we exploit the fact that $\Phi\big((u_3,u_1),(\beta_3,\beta_1)\big))$ 
satisfies the equation 
\begin{equation}\label{SchrEq}
\big(\partial_{u_1 u_1} - 2 \wp_{11}\big) \Phi = b\, \Phi,
\end{equation}
which is similar to a Schr\"{o}\-din\-ger type equation
\begin{equation}\label{SchrEqBl}
 \big(\partial_{zz} - \mathcal{U}(z)\big) \psi(z) = \mathcal{E} \psi(z)
\end{equation}
where $\mathcal{E}$ is a value of energy.

\begin{corollary}\label{C:BakerF}
 Suppose $\lambda(\wp(\alpha),\gamma)\big)\,{\in}\,\Lambda_1$ is defined by \eqref{ParamC25dX} 
 with $\wp(\alpha)\,{=}\,A$.
 Then for all $U_3\,{\in}\,\Complex$ the function
 \begin{gather}
 \psi(U_1) = \frac{\bm{\sigma}\big(B_3 - U_3,B_1 - U_1 +\tfrac{3}{5}\wp(\alpha) U_3;
 \lambda(\wp(\alpha),\gamma)\big)}
 {\bm{\sigma}\big(U_3,U_1+\tfrac{3}{5}\wp(\alpha) U_3;\lambda(\wp(\alpha),\gamma)\big)} 
 \,\rme^{U_1 \varpi}, \label{BackerFC25Deg}\\
\intertext{where $B_1$ is an arbitrary complex number,}
 B_3 = \frac{1}{\wp'(\alpha)}\bigg(2 \zeta(\alpha)  B_1
 + \log \frac{\sigma(\alpha-B_1)}{\sigma(\alpha+B_1)}\bigg), \notag \\
 \varpi = -\zeta(B_1) + \frac{1}{5} \wp(\alpha)  
  \bigg(1+\frac{18\zeta(\alpha) \wp(\alpha)}{5 \wp'(\alpha)}\bigg)B_1 
  + \frac{9 \wp(\alpha)^2}{25 \wp'(\alpha)} 
  \log \frac{\sigma(B_1-\alpha)}{\sigma(B_1+\alpha)}. \notag
\end{gather}
 satisfies the Schr\"{o}dinger equation
 \eqref{SchrEqBl} with the potential and energy 
\begin{equation}\label{PotenC25Deg}
 \mathcal{U}(U_1) = 2 \mathcal{S}^2  -2\wp(U_1)-2\wp(\alpha),\qquad \mathcal{E}=\wp(B_1),
\end{equation}
where the notation \eqref{Pnot} is used.
\end{corollary}
\begin{proof}
Consider the equation \eqref{SchrEq} with respect to 
the variable $U_1\,{=}\,u_1\,{-}\,\frac{3}{5}\wp(\alpha) u_3$, 
and use $\mathcal{P}_{11}$ from Corollary~\ref{C:JacInvPr} instead of $\wp_{11}$.
The equation acquires the form
\begin{equation*}
 \big(\partial_{U_1 U_1} - 2 \mathcal{P}_{11}\big) \Psi = b \Psi,
\end{equation*}
where $\Psi$ is obtained from $\Phi$ by applying the substitution \eqref{TransG12}
\begin{multline}\label{PsiFC25Deg}
 \Psi\big((U_3,U_1),(B_3,B_1)\big)) = 
 \frac{\bm{\sigma}\big(B_3 - U_3,B_1 - U_1+\tfrac{3}{5}\wp(\alpha) U_3;\lambda(\wp(\alpha),\gamma)\big)}
 {\bm{\sigma}\big(U_3,U_1+\tfrac{3}{5}\wp(\alpha) U_3;\lambda(\wp(\alpha),\gamma)\big)} \times \\ \times
 \exp\bigg({-}U_3 \int_{\infty}^{(b,y(b))} \dd R_3 - U_1 \int_{\infty}^{(b,y(b))}\dd R_1 \bigg),
\end{multline}
where $B_3\,{=}\,\beta_3$, $B_1\,{=}\,\beta_1\,{-}\,\tfrac{3}{5}\wp(\alpha) \beta_3$, 
the set of parameters $\lambda(\wp(\alpha),\gamma)$ is defined by \eqref{ParamC25dX} 
with $A\,{=}\,\wp(\alpha)$. Under the substitution \eqref{TransG12} we get 
\begin{equation*}
 B_3 = \int_\infty^{b+\tfrac{2}{5}\wp(\alpha)} \frac{\dd X}{-2(X-\wp(\alpha))Y(X)} = 
 \int_0^{B_1} \frac{\dd \xi}{\wp(\xi)-\wp(\alpha)}.
\end{equation*}
The factor $\exp\big({-}U_3 \int_{\infty}^{(b,y(b))} \dd R_3\big)$ 
is inessential so can be safely omitted.
Next, we compute
\begin{equation*}
  \dd R_1 = \bigg(\frac{\tfrac{1}{5}\wp(\alpha)}{-2Y} + \frac{X}{-2Y} 
  + \frac{\tfrac{9}{25}\wp(\alpha)^2}{-2Y(X-\wp(\alpha))}\bigg) \dd X
\end{equation*}
and obtain $\int_{\infty}^{(b,y(b))}\dd R_1 = \varpi$.
Finally, using \eqref{TransG12} we find $b\,{=}\,\wp(B_1)\,{-}\tfrac{2}{5}\wp(\alpha)$.
\end{proof}

\begin{remark}
The function $\mathcal{U}$ defined by \eqref{PotenC25Deg} satisfies the KdV equation
\begin{equation*}
 4\partial_{U_3}\mathcal{U} = \partial_{U_1}^3 \mathcal{U} 
 - 6 \mathcal{U} \partial_{U_1}\mathcal{U},
\end{equation*}
and is a stationary solution for higher equations of KdV hierarchy.
\end{remark}

Suppose, the roots $\{e_j\}_{j=1}^5$, $\sum_j e_j\,{=}\,0$, of polynomial 
$f(x,0)=x^5+\lambda_4 x^3+\lambda_6 x^2 +\lambda_8 x+\lambda_{10}$, 
that is branch points of the curve \eqref{CurveHG2}, are real numbers, 
and $e_1\,{\geqslant}\,e_2\,{\geqslant}\,e_3 \,{\geqslant}\, e_4 \,{\geqslant}\, e_5$.
Then the spectrum of operator in \eqref{SchrEq} is the union of three segments: 
$[e_5,\,e_4]\cup [e_3,\,e_2]\cup [e_1,\,\infty]$. When $\lambda\,{\in}\,\Lambda_1$
one of the segments, say $[e_5,\,e_4]$, contracts to produce a double point $A$.
Under the conditions we can interpret the results of Corollary~\ref{C:BakerF} in the following way.

\begin{corollary}
Let $(\omega,\,\omega')$ be periods of Weierstrass functions and assume 
$\Imw \omega = 0$, $\Rew \omega' = 0$.
Then, provided $\wp(\alpha)\in\Real$, formula \eqref{PotenC25Deg} defines 
one parametric families, with parameter $\varphi \in [-\tfrac{1}{2},\,\frac{1}{2}]$, 
of real-valued potentials in variable $x$ on real line  
\begin{align*}
  &\mathcal{V}_1(x) = \frac{1}{\omega^2} \mathcal{U}(\omega x),& &\text{with} \quad U_3 = 
  \frac{2\pi \imath}{\wp'(\alpha)}\varphi; &\\
  &\mathcal{V}_2(x) = \frac{1}{\omega^2} \mathcal{U}(\omega x+\tfrac{1}{2}\omega'),& 
  &\text{with} \quad U_3 =  \frac{2\pi \imath}{\wp'(\alpha)} \varphi 
  + \frac{1}{\wp'(\alpha)} \big(\zeta(\alpha) \omega' -  \alpha \eta'\big). &
\end{align*}
The operators $\partial_{xx} - \mathcal{V}_1(x)$ and $\partial_{xx} - \mathcal{V}_2(x)$ 
share a common spectrum 
\begin{equation*} 
 \big\{\wp(\alpha)\big\} \cup \big[\wp(\tfrac{1}{2}\omega'),\,
 \wp(\tfrac{1}{2}\omega+\tfrac{1}{2}\omega')\big] \cup 
 \big[\wp(\tfrac{1}{2}\omega),\infty\big].
\end{equation*}
\end{corollary}
\begin{proof}
Under the assumptions 
$\wp(z)$ is real when $z$ runs from the origin along the boundary of 
rectangle with sides $\tfrac{1}{2}\omega$ and $\tfrac{1}{2}\omega'$.
Further, both $(\wp(x),\,\wp'(x))$ and $(\wp(x+\tfrac{1}{2}\omega'),\,\wp'(x+\tfrac{1}{2}\omega'))$
are real for $x\in\Real$. Let $\alpha \in (0,\tfrac{1}{2}\omega)$, 
the functions $\mathcal{P}$ and $\mathcal{S}$ defined by
\eqref{SPnot} are real-valued. At $\alpha \in (\tfrac{1}{2}\omega + \tfrac{1}{2}\omega', \tfrac{1}{2}\omega')$
value of $\wp'(\alpha)$ is real, and $\mathcal{P}/\mathcal{P}^\ast=1$, as a result $\mathcal{S}$ is real.  
At $\alpha \in (\tfrac{1}{2}\omega',0)\cup (\tfrac{1}{2}\omega, \tfrac{1}{2}\omega + \tfrac{1}{2}\omega')$ 
values of $\wp'(\alpha)$ are imaginary, and $\mathcal{P}\mathcal{P}^\ast=1$ so $\mathcal{S}$ is imaginary.
\end{proof}

\begin{remark}
 The above potentials are unbounded except for  
 $\mathcal{V}_2(x)$ with $\varphi \in (-\tfrac{1}{2},\,\frac{1}{2})$ 
 in three cases: (1) $\Rew \alpha = 0$,
 (2) $\Rew \alpha = \omega$, (3) $\Imw \alpha = 0$.
\end{remark}

\subsection{Rank 3 lattices}\label{ss:3RankL}
Consider the space $\mathcal{C}$ of curves with a puncture at the common branch point at infinity.
Choose the following basis of holomorphic differentials
\begin{equation}
 h(x,y) = \big(1,\,x,\,{-}x^2,\,{-}(3x^3 + \lambda_4 x)\big)^t \frac{\dd x}{-2y},
\end{equation}
Denote by $\mathfrak{C}=\big(\mathfrak{a}_1$, $\mathfrak{a}_2$, 
$\mathfrak{b}_2$, $\mathfrak{b}_1\big)$  
a basis of homology cycles such that $\mathfrak{a}_i \circ \mathfrak{b}_j = \delta_{ij}$, see 
Figure~\ref{F:RootsHCycles}.
Denote by $\Omega$ a matrix of integrals of $h(x,y)$ over $\mathfrak{C}$, 
that is $\Omega = \int_{\mathfrak{C}} h(x,y)$. 
\begin{figure}[h]
 \centering 
 \includegraphics[width=0.6\textwidth]{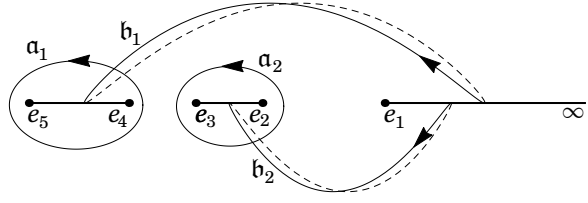}
 \caption{Sketch of branch points and basis homology cycles. \label{F:RootsHCycles}}
\end{figure}

If $\lambda\in\Lambda_2$, then $\rank \Omega = 4$ and $\Omega$ 
satisfies Legendre identity
\begin{equation}\label{LegId4}
 \Omega^t J \Omega = 2\pi \imath J
\end{equation}
for the symplectic matrix $J\,{=}\,\rm{codiag}(1,\,1,\,{-}1,\,{-}1)$.
First two rows of $\Omega$ generate a rank 4 lattice in $\Complex^2$, and thus 
define a two-dimensional complex torus as the quotient of $\Complex^2$ over the lattice. 
Meromorphic functions on the torus, that is \emph{four-periodic functions
on $\Complex^2$}, can be derived from sigma-function by taking logarithmic 
derivatives of order greater than 1. 
If $\lambda\in\Lambda_1 \cup \Lambda_0$, then $\rank \Omega < 4$. 

Introduce the notation
\begin{align*}
 &\mathcal{F}_k = \{\lambda\in\Complex^4 \mid \rank \Omega = k\},\qquad k=0,\,1,\,2,\,3,\,4.
\end{align*}
Evidently, the space $\Lambda\cong \Complex^4$ is a disjoint union $\Lambda=\cup_{k=0}^4 \mathcal{F}_k$
(cf. Proposition~\ref{P:LabmdaS}), 
where $\mathcal{F}_4=\Lambda_2$. Next,
\begin{lemma} 
$\mathcal{F}_3$ is the set of simple roots of the discriminant $\Delta(\lambda)$ of \eqref{CurveHG2}
\begin{equation*} 
 \mathcal{F}_3 = \{\lambda\mid \Delta(\lambda)=0, \partial_{\lambda} \Delta(\lambda)\neq 0\}.
\end{equation*}
\end{lemma}
\begin{proof}
Evidently, $\mathcal{F}_3\subset \Lambda_1$. 
In the case of $\lambda\,{\in}\,\Lambda_1$ we use the transformations \eqref{TransG12}
to obtain elliptic parametrization $(x,y)\,{=}\,\big(\wp(\xi)-\tfrac{2}{5}\wp(\alpha),\,
{-}\tfrac{1}{2}\wp'(\xi)\big(\wp(\xi)-\wp(\alpha)\big) \big)$ with the uniformizing parameter $\xi\in\Complex$.
Compute the integrals $I(x,y)=\int_\infty^{(x,y)} h(x,y)$ as functions in $\xi$
\begin{equation}\label{MerInt}
\begin{aligned}
 &I_1(\xi) = \frac{2\zeta(\alpha) }{\wp'(\alpha)} \xi + \frac{1}{\wp'(\alpha)}
 \log \frac{\sigma(\alpha-\xi)}{\sigma(\alpha+\xi)},\\
 &I_2(\xi) = \xi + \tfrac{3}{5}\wp(\alpha)I_1(\xi),\\
 &I_3(\xi) = \zeta(\xi) - \tfrac{6}{25} \wp(\alpha)^2 I_1(\xi)
 - \tfrac{1}{5}\wp(\alpha) I_2(\xi),\\
 &I_4(\xi) = -\tfrac{1}{2}\wp'(\xi) - \tfrac{3}{5}\wp(\alpha) 
 \big(\gamma_4 + \tfrac{12}{25}\wp(\alpha)^2\big) I_1(\xi) - \tfrac{9}{25}\wp(\alpha)^2 I_2(\xi)
 - \tfrac{3}{5}\wp(\alpha) I_3(\xi).
\end{aligned}
\end{equation}
Now we calculate the periods. Let $\Omega = \Big(\begin{smallmatrix} T_1 & T_2 & T_3 & T_4\\
H_1 & H_2 & H_3 & H_4 \end{smallmatrix}\Big)$, where $T_k$ and $H_k$ are 2-dimensional vectors. 
By taking expansion of $I(\xi)$ near $\xi = \alpha$ we find that 
\begin{equation*}
 \begin{pmatrix}  T_1 \\ H_1  \end{pmatrix} = 2\pi \imath \Res_{t=0} I(\alpha+t),\qquad
 \begin{pmatrix}  T_4 \\ H_4  \end{pmatrix} = \infty.
\end{equation*}
For this computations Figure~\ref{F:RootsHCyclesD} is instrumental.
\begin{figure}[h]
 \centering 
 \includegraphics[width=0.6\textwidth]{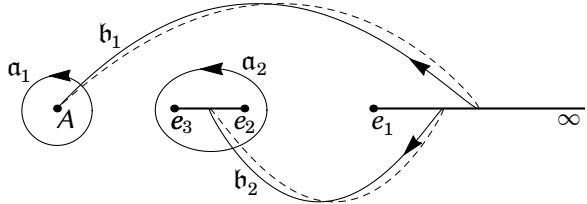}
 \caption{Sketch of branch points and basis homology cycles 
 when two branch points contract. \label{F:RootsHCyclesD}}
\end{figure}

On the other hand,
\begin{equation*}
 \begin{pmatrix}  T_2 \\ H_2  \end{pmatrix} = I(\xi+\omega) - I(\xi),\qquad
 \begin{pmatrix}  T_3 \\ H_3  \end{pmatrix} = I(\xi+\omega') - I(\xi).
\end{equation*}
Explicitly, for finite periods we have
\begin{gather}
 (T_1,T_2,T_3) = K_1 \begin{pmatrix} 0 & \omega & \omega' \\ 
 -\frac{\imath \pi}{\alpha} & \eta & \eta' \end{pmatrix}, \label{Tperiods} \\
 (H_1,H_2,H_3) = K_2 (T_1,T_2,T_3) +  K_3 \begin{pmatrix} 0 & \omega & \omega' \\ 
  0 & \eta & \eta' \end{pmatrix}, \label{Hperiods}
\end{gather}
where
\begin{gather*}
 K_1 = \begin{pmatrix} \frac{2}{\wp'(\alpha)}\zeta(\alpha) & -\frac{2}{\wp'(\alpha)}\alpha
 \\1 + \frac{6}{5}\frac{\wp(\alpha)}{\wp'(\alpha)}\zeta(\alpha) & 
 -\frac{6}{5}\frac{\wp(\alpha)}{\wp'(\alpha)}\alpha \end{pmatrix},\qquad
 K_2 = \begin{pmatrix} -\tfrac{9}{25}\wp(\alpha)^2 & 0 \\
 0 & -\big( \gamma_4 + \tfrac{12}{25}\wp(\alpha)^2\big) \end{pmatrix},\\
 K_3 = \begin{pmatrix} - \tfrac{1}{5}\wp(\alpha) & 1 \\
 \gamma_4 + \tfrac{6}{25} \wp(\alpha)^2 & -\frac{3}{5}\wp(\alpha) \end{pmatrix}.
\end{gather*}
and $\omega$, $\omega'$ are periods of the Weierstrass function $\wp$,
and $\eta = 2\zeta(\omega/2)$, $\eta' = 2\zeta(\omega'/2)$.
Thus, when $\lambda\in\Lambda_1$ $\rank \Omega = 3$ if and only if $\wp'(\alpha)\neq 0$.

To complete the proof it remains to notice that on $\Lambda_1$ 
the gradient of the discriminant \eqref{DeltaDef} vanishes together with $\wp'(\alpha)$. Indeed,
 \begin{equation}\label{DrvDelta}
 \partial_\lambda \Delta(\lambda)|_{\Lambda_1} = \frac{1}{5} (4\gamma_4^3 + 27 \gamma_6^2)\, 
 \big(\wp'(\alpha)\big)^6 \Big(\big(\tfrac{3}{5}\wp(\alpha)\big)^3, 
 \big(\tfrac{3}{5}\wp(\alpha)\big)^2, \tfrac{3}{5}\wp(\alpha), 1 \Big),
\end{equation}
here $\partial_\lambda$ stands for 
$(\partial_{\lambda_4},\,\partial_{\lambda_6},\,\partial_{\lambda_8},\,\partial_{\lambda_{10}})$.
By definition $4\gamma_4^3 + 27 \gamma_6^2$ does not vanish on $\Lambda_1$.
\end{proof}

Similarly, $\mathcal{F}_2$ is the set of double zeros of discriminant $\Delta(\lambda)$, 
and $\mathcal{F}_1$ is the set of triple zeros of discriminant $\Delta(\lambda)$:
\begin{align*}
 &\mathcal{F}_2 = \{\lambda \mid \partial_{\lambda} \Delta(\lambda)=0,\
\partial_{\lambda}^2 \Delta(\lambda)\neq 0\},\\
 &\mathcal{F}_1 = 
\{\lambda \mid \partial_{\lambda}^2 \Delta(\lambda)= 0,\ \partial_{\lambda}^3 \Delta(\lambda)\neq 0\}.
\end{align*}
Further, $\mathcal{F}_0$ is the set of $4$-tuple zeros of discriminant $\Delta(\lambda)$ which is
a singe point $\lambda\,{=}\,0$.

On the other hand, let $f(x)=x^5 + \sum_{k=0}^3 \lambda_{10-2 k} x^k$. Divisor of zeros $(f)_0$
is a formal product $p_1^{d_1} p_2^{d_2} \cdots p_5^{d_5}$, $p_i$ are distinct points, 
integers $d_i$ are non-negative and 
$\sum d_i = 5$ while $\sum d_i p_i = 0$. Assume $d_1\geqslant d_2 \geqslant \cdots
\geqslant d_5$ and denote $\deg (f)_0 = (d_1,d_2,\dots)$, for nonzero $d_i$. 
Clearly, $\deg (f)_0$ takes values in partitions of number $5$. 
Denote the dimension of corresponding subset of $\Lambda$ by $m$, we have 
$m = \# \deg (f)_0 -1$. In fact, $m$ equals the dimension of a component of $\mathcal{F}_m$.  
Table~\ref{T:Partitions} gives summary of all possible cases.
\begin{table}[h]
\caption{\label{T:Partitions}} 
\begin{tabular}{l||c|c|c|c|c|c|c}
\hhline {=#=|=|=|=|=|=|=} 
 $\deg (f)_0$ & $(1,1,1,1,1)$ & $(2,1,1,1)$ & $(3,1,1)$ &
 $(2,2,1)$ & $(3,2)$ & $(4,1)$ & $(5)$ \\
 genus $g$ & $2$ & $1$ & $1$ & $0$ & $0$ & $0$ & $0$\\
 $\# \deg (f)_0 -1$ & $4$ & $3$ & $2$ & $2$ & $1$ & $1$ & $0$ \\
 $\rank \Omega$ & $4$ & $3$ & $2$ & $2$ & $1$ & $1$ & $0$ 
\end{tabular} \\ 
\end{table}

This completes description of the stratification of $\Lambda$ by
the rank of corresponding lattice.

\begin{remark}
 Since $\mathcal{F}_2$ has nonempty intersections with both $\Lambda_1$ and $\Lambda_0$, 
 cf. Table~\ref{T:Partitions},
 the two-periodic functions on the associated `torus' can be of different nature: those that are essentially a combination
 of rational and elliptic functions, see Remark~\ref{R:SigmaF2}, and those that are combinations of
 exponential functions, see Theorem~\ref{Th:2}. The strata $\mathcal{F}_1$ and $\mathcal{F}_0$ are associated 
 with exponential and rational functions respectively.
\end{remark}

\subsection{Three-periodic functions}\label{ss:3PeriodF}
On the stratum $\mathcal{F}_3$ in the place of identity \eqref{LegId4} we have
\begin{gather*}
 \begin{pmatrix}
  T_1 & T_2 & T_3 \\ H_1 & H_2 & H_3 
 \end{pmatrix}^t J \begin{pmatrix}
  T_1 & T_2 & T_3 \\ H_1 & H_2 & H_3 
 \end{pmatrix} =
 2\pi \imath \begin{pmatrix} 0 & 0 & 0 \\ 0 & 0 & 1 \\ 0 & -1 & 0 \end{pmatrix}.
\end{gather*}

\begin{corollary}\label{C:PerProp}
For all $u=(u_3,\,u_1)\in\Complex^2$ 
 sigma-function $\bm{\sigma}(u;\lambda)$ obeys the periodicity property
\begin{gather*}
 \frac{\bm{\sigma}(u \pm T_k;\lambda)}
 {\bm{\sigma}(u;\lambda)}\bigg|_{\lambda\in\mathcal{F}_3} = 
 -\exp \big\{\pm H_k^t \left( \begin{smallmatrix} 0 & 1 \\ 1 & 0 \end{smallmatrix} \right) 
 (u \pm \tfrac{1}{2} T_k)\big\},\qquad k=1,\,2,\,3.
\end{gather*}
\end{corollary}
Proof follows directly from the periodicity property of genus $2$ sigma-function.

\begin{remark}
The function $\Phi(u,\beta)$ defined by \eqref{BackerFC25}
 has Bloch property on $\mathcal{F}_4=\Lambda_2$ and keeps the property 
 when restricted to $\mathcal{F}_3$ 
\begin{equation*}
 \Phi(u +T_i,\beta) = \Phi(u,\beta) \rme^{M_i T_i},\qquad i=1,2,3.
\end{equation*}
The ``quasi-momenta'' $M_i$ are given by rather cumbersome expressions, which, however,
can be readily deduced from \eqref{MerInt}, \eqref{Tperiods} and \eqref{Hperiods} in a condensed form.
Let $\beta^t = (I_1(\alpha),I_2(\alpha))$ and $\rho =(I_4(\alpha),I_3(\alpha))$, 
where $\alpha$ is the image of double point~$A$, that is $\wp(\alpha)=A$. Then we have
\begin{gather*}
 M_1 = \rho + \beta^t K_2,\qquad
 M_2 = M_3 = \rho  +  \beta^t (K_2 + K_3 K_1^{-1}).
\end{gather*}
\end{remark}

Now, return to discussing three-periodic functions.
Over $\mathcal{F}_3$ any order greater than 1 logarithmic derivative of sigma-function 
is a three-periodic function.

Introduce the function
\begin{equation}\label{Pfunct}
 \mathcal{P}(u_3,u_1) = \frac{\sigma \big(\alpha + u_1 - \tfrac{3}{5} \wp(\alpha) u_3 \big)}
{\sigma \big(\alpha - u_1 + \tfrac{3}{5} \wp(\alpha) u_3\big)}\, 
\rme^{ \big(\wp'(\alpha) + \tfrac{6}{5} \wp(\alpha)\zeta (\alpha)\big) u_3 - 2\zeta (\alpha) u_1}
\end{equation}
with $\wp'(\alpha)\neq 0$. It is straightforward to verify that $T_1$, $T_2$, $T_3$ 
are periods of the function $\mathcal{P}(u_3,u_1)$. 

\begin{corollary}\label{C:3periodF}
Any meromorphic three-periodic function  in two variables $(u_3,\,u_1)$ with the periods 
$T_1$, $T_2$, $T_3$ 
is a rational function of 
$\bm{\mathcal{P}}_{\textup{basis}}=\big(\mathcal{P}(u_3,u_1),\,\wp(u_1 - \tfrac{3}{5} 
\wp(\alpha) u_3)$,
$\wp'(u_1 - \tfrac{3}{5} \wp(\alpha) u_3)$, $\wp(\alpha)$, $\wp'(\alpha)\big)$.
\end{corollary}
\begin{proof}
Any genus 2 Abelian function, that is a meromorphic four-periodic function of $(u_3,u_1)$, has a unique 
representation as the rational function of 
$\bm{\wp}_{\textup{basis}}=(\wp_{11}$, $\wp_{13}$, $\wp_{111}$, $\wp_{113}$,
$\wp_{1111}$, $\wp_{1113})$, in particular
\begin{equation}\label{lambdaWP}
\begin{aligned}
 &\lambda _4 = \tfrac{1}{2}\wp _{1111} - 3 \wp _{11}^2 - 2 \wp _{13},\\
 &\lambda _6 = \tfrac{1}{2}\wp _{1113} - \tfrac{1}{2}\wp _{1111} \wp_{11} 
 + \tfrac{1}{4}\wp_{111}^2 + 2 \wp_{11}^3 - 2 \wp_{13}\wp_{11},\\
 &\lambda _8 = - \tfrac{1}{2}\wp _{1113} \wp_{11} - \tfrac{1}{2}\wp _{1111} \wp_{13}
 +\tfrac{1}{2} \wp_{113}\wp_{111} + \wp_{13}^2 + 4\wp_{11}^2 \wp_{13},\\
 &\lambda _{10} = - \tfrac{1}{2}\wp _{1113} \wp_{13} + \tfrac{1}{4}\wp_{113}^2 +
 2 \wp_{13}^2 \wp_{11}.
\end{aligned}
\end{equation}
The composition of $\Delta(\lambda)$, see \eqref{DeltaDef}, with \eqref{lambdaWP} defines a polynomial 
$\Delta\big(\lambda(\bm{\wp}_{\text{basis}})\big)$. When $\lambda\,{\in}\,\Lambda_1$,
the polynomial $\Delta\big(\lambda(\bm{\wp}_{\text{basis}})\big)$ should vanish, while 
$\delta(\gamma)\,{=}\,4\gamma_4^3 + 27\gamma_6^2$ should be nonzero.
Taking into account
\begin{align*}
 &\gamma_4 = \frac{\wp'(U_1)^2 - 4 \wp(U_1)^3 - \big(\wp'(\alpha)^2 - 4\wp(\alpha)^3\big)}
 {4\big(\wp(U_1)-\wp(\alpha)\big)},\\
 &\gamma_6 = 
 -\frac{\wp(\alpha) \big(\wp'(U_1)^2 - 4 \wp(U_1)^3\big) - 
 \wp(U_1) \big(\wp'(\alpha)^2 - 4\wp(\alpha)^3\big)}{4\big(\wp(U_1)-\wp(\alpha)\big)},
\end{align*}
where $U_1\,{=}\,u_1 \,{-}\, \tfrac{3}{5} \wp(\alpha) u_3$,
we see the condition $4\gamma_4^3 + 27\gamma_6^2\,{\neq}\,0$ turns into a condition on a polynomial in
$\wp(U_1)$, $\wp'(U_1)$, $\wp(\alpha)$, and $\wp'(\alpha)$.
By \eqref{IJPrwpG2}--\eqref{IJPrwpG2dG} we can express $\wp_{11}$, $\wp_{13}$, $\wp_{111}$, $\wp_{113}$
as rational functions of $\bm{\mathcal{P}}_{\text{basis}}$.
Differentiating expressions for $\wp_{111}$, $\wp_{113}$ with respect to $u_1$ we get
the rational functions
\begin{equation*}
\begin{split}
 \wp_{1111} =\,&\, 6\mathcal{S}^4 - 8 \big(2\wp(U_1)+\wp(\alpha)\big)\mathcal{S}^2
 +4 \wp'(U_1)\mathcal{S} \\ & + 4\wp(U_1)^2 + 10 \wp(\alpha)\wp(U_1) + 
 4\wp(\alpha)^2 - \frac{\wp'(U_1)^2 - \wp'(\alpha)^2}{2\big(\wp(U_1)-\wp(\alpha)\big)},\\
 \wp_{1113} =\,&\,{-} 6\big(\wp(U_1)-\tfrac{2}{5}\wp(\alpha)\big) \mathcal{S}^4
 + 6\wp'(U_1) \mathcal{S}^3  \\ & + \bigg(4\wp(U_1)^2 + \tfrac{38}{5} \wp(\alpha)\wp(U_1) 
 + \tfrac{14}{5}\wp(\alpha)^2 - 
 \frac{3\big(\wp'(U_1)^2 - \wp'(\alpha)^2\big)}{2\big(\wp(U_1)-\wp(\alpha)\big)}\bigg) \mathcal{S}^2 
  \\ & - \tfrac{2}{5}\wp'(U_1) \Big(10\wp(U_1) + 11\wp(\alpha)\Big) \mathcal{S} 
 + 6 \wp(\alpha) \big(\tfrac{2}{5}\wp(U_1)^2 + \wp(\alpha)\wp(U_1) + 
 \tfrac{2}{5}\wp(\alpha)^2 \big) \\ & - \tfrac{4}{5} \big(\wp(U_1)^2 + \wp(\alpha)^2\big)
 + \frac{9\wp(U_1)\big(\wp'(U_1)^2 + \wp'(\alpha)^2\big)}{5\big(\wp(U_1)-\wp(\alpha)\big)}.
\end{split}
\end{equation*}
These rational expressions for $\bm{\wp}_{\text{basis}}$ substituted in 
$\Delta\big(\lambda(\bm{\wp}_{\text{basis}})\big)$
make it vanish identically. Furthermore, we can re-express a rational function of $\bm{\wp}_{\text{basis}}$
as a rational function of $\bm{\mathcal{P}}_{\text{basis}}$.
\end{proof}

The parametrization of $\bm{\wp}_{\text{basis}}$ by rational functions of $\bm{\mathcal{P}}_{\text{basis}}$
is analogous to the parametrization of $\lambda$ in terms of $a_2$ and $\gamma$, cf. \eqref{CoefC25degSubs}.
In fact, the former parametrization is induced by the latter, which is clearly seen if we follow the connection
between sigma-function $\bm{\sigma}$ and generators $\bm{\wp}_{\text{basis}}$ of the field of fiber-wise Abelian functions 
on the universal space of genus $2$ Jacobi varieties.

\begin{remark}
Note that the function $f(z_1,z_2) = \mathcal{P}(z_1/\wp'(\alpha),
c z_2+\tfrac{3}{5}\wp(\alpha)z_1/\wp'(\alpha))$, where $c\neq 0$ is an arbitrary number, 
is a solution of the following system of functional equations
\begin{equation*}
 f(z_1,z_2) f(z_1,-z_2) = \exp(z_1), \qquad f(z_1,z_2) = -f(-z_1,-z_2).
\end{equation*}
\end{remark}

\begin{proposition}
 A field of three-periodic functions is a transcendental extension of 
 the field of elliptic functions with transcendence degree 1.
\end{proposition}
Proof follows from Corollary~\ref{C:3periodF}, the 
function $\mathcal{P}(u_3,u_1)$ serves as the transcendental element.

\section{Concluding remarks}
For all values of parameters $\lambda$  sigma-function $\bm{\sigma}(u;\lambda)$
is essentially a function of the same nature, that is remains holomorphic and entire in all its
arguments for singular curves as well. It possesses a periodicity property, in particular, 
for the case of actual genus $1$ and $\partial_{\lambda} \Delta \neq 0$ it is given
in Corollary~\ref{C:PerProp}, and the functions $-\partial_{u_i} \partial_{u_j} \log \bm{\sigma}(u;\lambda)$
are three-periodic. The whole classification of degeneration strata is given in Table~$1$. 
`Degenerate' expressions \eqref{SolG2} 
and \eqref{SolG2d}, at special values
of parameters $\lambda$, are useful for solving generalized Jacobi inversion problem, and
Schr\"{o}dinger equation with periodic potential.

The technique we use above can be extended almost literally to higher genera
hyperelliptic sigma-functions. Generalization to non-hyperelliptic sigma-functions is a challenging problem.
Based on \eqref{SolG2} and well-known formula for degenerate Weierstrass sigma-function 
\cite{BatErd}, namely when $(g_2,g_3)\mapsto (12a^2,-8a^3)$
\begin{equation*}
 \sigma(u) \mapsto \tfrac{1}{2} \sqrt{3a}\, \rme^{-\frac{1}{2}a u^2} 
 \big(\rme^{\sqrt{3a}u} - \rme^{-\sqrt{3a}u} \big),
\end{equation*}
we conjecture that evaluation of a genus $g$ hyperelliptic sigma-function at a stratum 
of parameters $\Lambda_{g-1}$,
where genus of the underlying curve falls by $1$, has similar structure 
\begin{equation*}
 \bm{\sigma}(u) \mapsto C \rme^{- u^t \mathcal{Q} u} 
 \Big(\sigma(\mathcal{A}+u) \rme^{\mathcal{M}^t u} 
 - \sigma(\mathcal{A}-u) \rme^{-\mathcal{M}^t u} \Big).
\end{equation*}
Here sigma-function on the left hand side is in genus $g$, while
sigma-function on the right hand side is in genus $g-1$, 
scalar $C$, $g\times g$ matrix $\mathcal{Q}$ and vectors $\mathcal{A}$ and $\mathcal{M}$
are expressed with the hep of first and second kind Abel integrals as 
functions of the coordinates of a double point and the parameters of genus $g-1$ curve
corresponding to a point in $\Lambda_{g-1}$. From this viewpoint sigma-function
in genus $0$ is a constant function, say,~$1$. We can regard the result of degeneration 
as an action of an operator $\mathcal{T}$, which is in essence an evaluation operator. 
Then properly tuned operators $\mathcal{T}(a)$
and $\mathcal{T}(b)$ associated with double points at $a$ and $b$ 
commute with respect to composition, which opens a possibility to study 
further degeneration of sigma-function in a more abstract setting. 

For the generalized Jacobi inversion problem considered in Subsection~\ref{ss:GJIP}
an alternative solution is known within the framework of the generalized Theta-function theory,
which is developed by E. Previato \cite{Previato}, Yu. Fedorov \cite{Fedorov}, 
H. Braden and Yu. Fedorov \cite{BraFed}.
A connection between the degenerate sigma-function \eqref{SolG2} and the 
generalized Theta-function can be traced through the relation between sigma- and theta-functions
in genus $1$, see \cite{BatErd}.

The subject of Subsection~\ref{ss:SchEq} may be viewed as the simplest examples of a potential
of mixed solitonic and finite-gap nature. It is of considerable interest to explicitly construct potentials
that possess arbitrary collection of points and segments in the place of spectra.

In general, lattices of odd ranks lead to generalized Jacobi varieties, 
see \cite{Fay,Previato}. The rank three lattice from Subsection~\ref{ss:3RankL} 
is an example of that. The corresponding generalized Jacobi variety is a union of
a cylinder and a torus. At the same time, we conjecture that a field of $2g+1$-periodic functions 
can be effectively constructed as a transcendental extension of the field of hyperelliptic 
Abelian functions in genus $g$ with help of a single transcendental element 
of a form similar to \eqref{Pfunct}, namely
\begin{equation*}
 \mathcal{P}(u_{g+1},u) = \frac{\sigma(\alpha+u)}{\sigma(\alpha-u)} \exp\big\{
 c(\alpha) u_{g+1} + d(\alpha)^t u\big\},
\end{equation*}
where $\sigma$ denotes genus $g$ sigma-function, $u,\, \alpha \in\Complex^{g}$, $u_{g+1}\in\Complex$, 
and $c(\alpha)$, $d(\alpha)$ are appropriate functions.

Our study of three-periodic functions of two complex variables 
will be extended in our future publications, in particular we plan to derive 
explicit form of addition law and to find special dynamical systems solvable by these functions.

\subsection*{Acknowledgements}
This research was supported in part by the grant GMJT 2014-20 of the Glasgow Learning, 
Teaching and Research Fund.
One of the authors (JB) thanks professors J.\;C. Eilbeck and Fredericke van Wjick for warm hospitality
during her stay in Edinburgh in winter 2015. We also thank
professors C. Athorne, H. W. Braden, J.\;C. Eilbeck, Victor Enolski for stimulating discussions.


\begin{thebibliography}{9}
\bibitem{Klein} Klein F. 
\"{U}ber hyperelliptische Sigmafunktionen. (Erster Aufsatz), 
Gesammelte mathematische Abhandlungen, Band 3, Springer, Berlin, 1923, P.323 

\bibitem{Weier1894} Weierstrass K. Zur Theorie der elliptischen Funktionen, Mathematische Werke, Bd. 2,
Berlin, Teubner, 1894, pp.\,245--255

\bibitem{Baker} Baker H. F. Multiply periodic functions, CUP, 1907.

\bibitem{BEL1997} Buchstaber V. M., Enolskii V. Z., and Leykin D. V. Hyperelliptic Kleinian
functions and applications, volume 179, pages 1–34. Advances in Math. Sciences, 
AMS Translations, series 2, Moscow State University and University
of Maryland, College Park, 1997.

\bibitem{Nak2010} Nakayashiki A. On algebraic expressions of sigma functions for $(n,s)$-curves, 
Asian J. Math., Vol. 14, No 2, pp. 175--212, 2010.

\bibitem{Nak2015} Nakayashiki A. Tau function approach to theta functions, 
IMRN, rnv297, 2015.

\bibitem{KorShr} Korotkin D., Shramchenko V. On higher genus Weierstrass sigma-function, 
Physica D: Nonlinear Phenomena, \textbf{241} 2012, 2083-2284

\bibitem{EEE2013} Eilbeck J. C., Eilers K. and Enolski V.Z., 
\emph{Periods of second kind differential of $(n,s)$-curves}, 
Proceedings of Moscow Mathematical Society 74, 297-315, 2013.

\bibitem{BEL1999}  Buchstaber V. M.,  Enolskii V. Z., and  Leykin D. V. 
\textsl{Rational Analogs of Abelian Functions}, Functional Analysis and Its Applications, 
Vol. 33, No. 2, 1999

\bibitem{BL2002}  Buchstaber V. M., and Leykin D. V. \textsl{Polynomial Lie Algebras}, 
Functional Analysis and Its Applications, Vol. 36, No. 4, pp. 267--280, 2002

\bibitem{BL2004}  Buchstaber V. M., and Leykin D. V. {Heat Equations in a Nonholonomic Frame},
Functional Analysis and Its Applications, Vol. 38, No. 2, pp. 88--101, 2004

\bibitem{BL2008} Buchstaber V. M., and Leykin D. V. 
\textsl{Solution of the Problem of Differentiation of Abelian Functions
over Parameters for Families of (n, s)-Curves}, 
Functional Analysis and Its Applications, Vol. 42, No. 4, pp. 268--278, 2008

\bibitem{AEE2003} Athorne C., Eilbeck J. C., Enolski V. 
Identities for the classical genus two $\wp$ function, 
J. Geom. Phys., Vol. 48 (2003), pp. 354--368.

\bibitem{A2008} Athorne C. Identities for hyperelliptic $\wp$-functions of genus one,
two and three in covariant form, J. Phys. A: Math. Theor. 41 (2008)  415202.

\bibitem{A2011} Athorne C. A generalization of Baker’s quadratic formulae for hyperelliptic 
$\wp$-functions, Physics Letters A 375 (2011) 2689--2693.

\bibitem{A1999} Athorne C. Algebraic invariants and generalized Hirota derivatives,
Physics Letters A 256 (1999) 20--24.

\bibitem{EA2012} England M., Athorne C. Building Abelian Functions
with Generalised Baker–Hirota Operators,  SIGMA 8 (2012), 037, 36 pages

\bibitem{BG2006} S. Baldwin, J. Gibbons, Genus 4 trigonal reduction of the Benney equations, 
J. Phys. A: Math. Theor. 39 (2006) 3607--3639.

\bibitem{EEMOP2007} Eilbeck J. C., Enolski V., Matsutani S., \^{O}nishi Y,
Previato E. Abelian Functions for Trigonal Curves of Genus Three, IMRN, rnv140, 2007.

\bibitem{BEGO2008} Baldwin S., Eilbeck J. C., Gibbons J., \^{O}nishi Y,
Abelian functions for cyclic trigonal curves of genus 4, 
J. Geom. Phys., Vol. 58 (2008), pp. 450--467.

\bibitem{EEG2009}  Eilbeck J. C., England M., 
Abelian functions associated with a cyclic tetragonal curve of genus six,
J. Phys A, Vol. 42, No 9, 095210.

\bibitem{EEG2010} Eilbeck J. C., Enolski V., Gibbons J. 
Sigma, tau and abelian functions of agebraic curves, 2010 J. Phys. A: Math. Theor. 43 455216	

\bibitem{BatErd} Bateman H., Erd\'{e}lyi A. {Higher transcendental functions} V.\,2, 1953

\bibitem{BraFed} Braden H.W.,  Fedorov Yu. N. An extended Abel-Jacobi map, J. Geom. Phys \textbf{58} (2008),
P. 1346--1354.

\bibitem{Fay} Fay J.D., Theta functions on Riemann surfaces,
LNM 352, Springer, Berlin, 1973.

\bibitem{Fedorov} Fedorov Yu. Classical Integrable Systems and Billiards Related to Generalized Jacobians,
Acta Appl. Math. \textbf{55} (1999), P. 251--301

\bibitem{Gaudin} Gaudin M. La fonction d'onde de Bethe, Masson, Paris, 1983.

\bibitem{G1980}  Givental’ A. B. {Displacement of invariants of groups 
that are generated by reflections and are connected with simple singularities of functions}
Functional Analysis and Its Applications, 1980, \textbf{14}:2, 81--89.

\bibitem{Previato} Previato E. Hyperelliptic Quasi-Periodic and Soliton Solutions of the 
Nonlinear Schr\"{o}dinger Equation, Duke Math J., \textbf{52}:2 (1985), P.329--377.

\bibitem{Z1976}  Zakalyukin V. M.
\textsl{Reconstructions of wave fronts depending on one parameter}, 
Functional Analysis and Its Applications, 1976, \textbf{10}:2, 139--140.

\end{thebibliography}
\end{document}